%% file: report.tex
\documentclass[letterpaper,twoside,12pt]{report}
\usepackage{graphicx,amssymb,amstext,amsmath,amsthm,makeidx,verbatim,latexsym}  
\usepackage{pstricks,pst-node,pst-tree}
\usepackage{fancyhdr}

\newtheorem{lemma}{Lemma}[chapter]
\newtheorem{theorem}[lemma]{Theorem}

\newtheorem{proposition}[lemma]{Proposition}
\newtheorem{corollary}[lemma]{Corollary}
\newtheorem{conjecture}[lemma]{Conjecture}

\theoremstyle{definition}
\newtheorem{definition}[lemma]{Definition}
\theoremstyle{definition}
\newtheorem{example}{Example}[chapter]

\theoremstyle{remark}
\newtheorem*{remark}{Remark}

\makeindex

\begin{document}


\include{TitlePage}


\pagenumbering{roman}

\include{abstract}
\include{resume}


\include{acknowledgements}


\tableofcontents


\include{chap1}
\include{chap2}
\include{chap3}

\include{chap4}
\include{chap5}
\appendix
\include{app}

\addcontentsline{toc}{chapter}{Bibliography}
\bibliographystyle{alpha}
\bibliography{mybib}


\printindex


\end{document}

%% file: TitlePage.tex
\begin{titlepage}
\begin{center}
\huge Non-Deterministic Communication Complexity of Regular Languages\\
\vfill
\large Anil Ada \\
\vfill
\textit{School of Computer Science\\
McGill University, Montr\'eal\\}
January, 2008\\
\vfill
A thesis submitted to the Faculty of Graduate Studies and Research in partial fulfillment of the requirements of the degree of Master of Science.\\
\vfill
Copyright \copyright Anil Ada 2007.
\end{center}
\vspace*{0.5 in}
\end{titlepage}

%% file: abstract.tex
\pagenumbering{roman} \setcounter{page}{1}
\thispagestyle{plain}
\chapter*{Abstract\markboth{Abstract}{Abstract}}

The notion of communication complexity was introduced by Yao in his seminal paper \cite{Yao79}. In \cite{BFS86}, Babai Frankl and Simon developed a rich structure of communication complexity classes to understand the relationships between various models of communication complexity. This made it apparent that communication complexity was a self-contained mini-world within complexity theory. In this thesis, we study the place of regular languages within this mini-world. In particular, we are interested in the non-deterministic communication complexity of regular languages.

We show that a regular language has either $O(1)$ or $\Omega(\log n)$ non-determi- nistic complexity. We obtain several linear lower bound results which cover a wide range of regular languages having linear non-deterministic complexity. These lower bound results also imply a result in semigroup theory: we obtain sufficient conditions for not being in the positive variety $Pol(\mathcal{C}\it{om})$.

To obtain our results, we use algebraic techniques. In the study of regular languages, the algebraic point of view pioneered by Eilenberg (\cite{Eil74}) has led to many interesting results. Viewing a semigroup as a computational device that recognizes languages has proven to be prolific from both semigroup theory and formal languages perspectives. In this thesis, we provide further instances of such mutualism.

%% file: resume.tex
\chapter*{R\'esum\'e}

La notion de complexit\'e de communication a d'abord \'et\'e introduite 
par Yao \cite{Yao79}. Les travaux fondateurs de Babai et al. \cite{BFS86} ont 
d\'evoil\'e une riche structures de classes de complexit\'e de 
communication qui permettent de mieux comprendre la puissance de divers 
modèles de complexit\'e de communication. Ces r\'esultats ont fait de la 
complexit\'e de communication une sorte de maquette à petite échelle du 
monde de la complexité. Dans ce m\'emoire, nous \'etudions la place des 
langages r\'eguliers dans cette maquette. Plus pr\'ecis\'ement, nous 
chercherons à d\'eterminer la complexit\'e de communication 
non-d\'eterministe de ces langages.

Nous montrons qu'un langage r\'egulier a une complexit\'e de 
communication soit $O(\log n)$, soit $\Omega(\log n)$. Nous 
\'etablissons de plus des bornes inf\'erieures lin\'eaires sur la 
complexit\'e non-d\'eterministe d'une vaste classe de langages. 
Celles-ci fournissent \'egalement des conditions suffisantes pour qu'un 
langage donn\'e n'appartienne pas à la vari\'et\'e positive $Pol(\mathcal{C}\it{om})$.

Nos r\'esultats se basent sur des techniques alg\'ebriques. Dans 
l'\'etude des langages r\'eguliers, le point de vue alg\'ebrique, 
d\'evelopp\'e initialement par Eilenberg \cite{Eil74} s'est r\'ev\'el\'e 
comme un outil central.  En effet, on peut voir un semigroupe fini comme 
une machine capable de reconna\^\i tre des langages et cette perspective 
a permis des avanc\'ees tant en th\'eorie des semigroupes qu'en 
th\'eorie des langages formels. Dans ce m\'emoire, nous \'etablissons de 
nouveaux exemples de ce mutualisme.

%% file: acknowledgements.tex
\chapter*{Acknowledgments}

First, I would like to express my deepest gratitude to Prof. Denis Th\'{e}rien. I thank him for trusting me and accepting me as his student. His enthusiasm for complexity theory and mathematics was highly contagious and is one of the reasons I am in this field. I also thank him for supporting me financially.

I am indebted to my co-supervisor Prof. Pascal Tesson for many things. I thank him for introducing me to the subject of this thesis. I am grateful for the extremely useful discussions we had, which taught me a lot of things. I also thank him for his constructive comments on the earlier drafts of the thesis.

I am very fortunate to have met Arkadev Chattopadhyay, L\'{a}szl\'{o} Egri, Navin Goyal and Mark Mercer in the complexity theory group. Thanks to my office mate L\'{a}szl\'{o} for discussions about mathematics and many other topics. I thank Arkadev for discussing the thesis problem with me and for sharing his keen insight on various things. Thanks to Navin and Mark for generously sharing their knowledge. I have learned a lot from all of them.

I thank the academic and administrative staff of the computer science department. I have met many wonderful people over the years and I feel privileged to be a part of this family.

I would also like to thank all my friends in Montr\'{e}al for making life fun for me here.

\hfill

Finally, biggest thanks go to my parents. Their love and support never wavered and this has made everything possible.

%% file: chap1.tex
\pagenumbering{arabic}

\pagestyle{fancy}                       
\fancyfoot{}                            
\renewcommand{\chaptermark}[1]{         
\markboth{\MakeUppercase{\chaptername\ } \thechapter.\ #1}{}}
\renewcommand{\sectionmark}[1]{        
\markright{\thesection.\ #1}}
\fancyhead[LO,RE]{\thepage}    
\fancyhead[LE]{\rightmark}      
\fancyhead[RO]{\leftmark}     

\chapter{Introduction}\label{chapter1}

\section{Computational Complexity Theory}

The theory of computation is one of the fundamental branches of computer science that is concerned with the computability and complexity of problems in different computational models. Computability theory focuses on the question of whether a problem can be solved in a certain computational model. On the other hand, complexity theory seeks to determine how much resource is sufficient and necessary for a computable problem to be solved in a computational model.

Simply put, a \emph{computing device} is a machine that performs calculations automatically: it can be as complicated as a personal computer and as simple as an automatic door. In theoretical computer science, a \emph{computational model} is a pure mathematical definition which models a real-world computing device. This abstraction is necessary in order to rigorously study computation, its power and limitations. The most studied computational model (which physically corresponds to the everyday computers we use) is the Turing Machine. The most studied resources are time and space (memory) measured with respect to the input size. Based on these resources, different complexity classes can be defined. For instance, $P$ and $NP$ are classes of problems that can be solved in polynomial (in the input size) time using a deterministic and a non-deterministic Turing Machine respectively. Whether these classes are equal or not is without a doubt one of the biggest open questions in computer science and mathematics.

There are various interesting computational models including (but not limited to) Turing Machines, finite automata, context-free grammars, boolean circuits and branching programs. Their countless applications span computer science. For instance, when designing a new programming language one would find grammars useful. Finite automata and regular languages have applications in string searching and pattern matching. When trying to come up with an efficient algorithm, the theory of NP-completeness can be insightful. Many cryptographic protocols rely on theoretical principles. All these applications aside, the mathematical elegance and aesthetic inherent in theory of computation is enough to attract many minds around the world. And perhaps the main reason that computer science is called a ``science" is because of the study of theoretical foundations of computer science.

Despite the intense efforts of many researchers, our understanding and knowledge of computational complexity is quite limited. Similar to the $P$ versus $NP$ question, there are many other core questions (in different computational models) that beg to be answered. The focus of research in complexity theory is twofold. Given a certain problem, a computational model and a resource:
\begin{itemize}
\item What is the maximum amount of resource we need to solve the problem in the computational model?
\item What is the minimum amount of resource we need to solve the problem in the computational model?
\end{itemize}
The ultimate goal is to find matching upper and lower bounds. The first question can be answered by depicting a method\footnote{In the Turing Machine computational model, the method is called an \emph{algorithm}.} of solving the problem and analyzing the amount of resource this method consumes. Almost always, the more challenging question is the second one. Proving results of the form ``Problem $p$ requires at least $x$ resource." requires us to argue against all possible methods that solve the problem. In most computational models, this is intrinsically hard. Yet it should be also noted that complexity theory is a relatively new field and therefore can be considered as an amenable discipline of mathematics.

\section{Communication Complexity}

In this thesis, we will be studying a computational model which emulates distributed computing: communication protocols. In this model, there are usually two computers that are trying to collaboratively evaluate the value of a given function. The difficulty is that the input is distributed among the two computers in a predetermined adversarial way so that neither computer can evaluate the value of the function by itself. Therefore, in order to determine the value of the function, these computers need to communicate over a network. The communication will be carried out according to a protocol that has been agreed upon beforehand. The resource we are interested in is the number of bits that is communicated i.e. we would like to determine the \emph{communication complexity} of a given function.

As an example, consider two files that reside in two computers. Suppose we wanted to know if these two files were copies of each other. How many bits would the computers need to communicate in order to conclude that the files are the same or not? What is the best protocol for the computers to accomplish this task?

Note that the scenario here is quite different from \emph{information theory}. In information theory, the goal is to robustly transmit a predetermined message through a noisy channel and there is no function to be computed. In the communication complexity setting, the channel of communication is not noisy. What is sent through the channel is determined by the protocol and it usually changes according to the inputs of the computers and the communication history.

There are various models for communication complexity. The first defined was the 2-player deterministic model. Since then, non-deterministic, randomized, multi-party, distributional, simultaneous and many more models have been defined and analyzed.

Although the mathematical theory of communication complexity was first introduced in light of its applications to parallel computers (\cite{Yao79}), it has been shown to have many more applications where the need for communication is not explicit. These applications include time/space lower bounds for VLSI chips (\cite{KN97}), time/space tradeoffs for Turing Machines (\cite{BNS92}), data structures (\cite{KN97}), boolean circuit lower bounds (\cite{Gro92}, \cite{HG91},\cite{Nis93},\cite{RM97}), pseudorandomness (\cite{BNS92}), separation of proof systems (\cite{BPS07}) and lower bounds on the size of polytopes representing $NP$-complete problems (\cite{Yan91}).

\section{Algebraic Automata Theory}

One of the fundamental (and simplest) computational models is the finite automaton and it is usually the first model in theory of computation that computer science students are introduced to. The word ``finite" refers to the memory of the machine and finite automata are models for computers with an extremely limited amount of memory (for example an automatic door). Even though it is a quite limited model, its well-known applications include text processing, compilers and hardware design.

In a nutshell, finite automata are abstract machines such that given a word over some alphabet as an input, it either accepts or rejects the word after processing each letter of the word sequentially. The set of all words that a finite automaton accepts is called the \emph{language} corresponding to the finite automaton and we say that the language is \emph{recognized} by this automaton. A language recognized by some finite automaton is called a \emph{regular language}.

Algebra has always been an important tool in the study of computational complexity. In the study of regular languages, semigroup theory\footnote{A semigroup is a set equipped with a binary associative operation.} has been the dominant tool. It should be mentioned that semigroups have shed new light not only on regular languages but on computational theory in general. On top of this, it is also true that computational theory has led to advances in the study of semigroup theory (\cite{TT04}).

The link between semigroups and regular languages has been established by viewing a semigroup as a computational machine that accepts/rejects words over some alphabet. In this context, it is not difficult to prove that the family of languages that finite semigroups recognize is exactly the regular languages. In fact, the connection between finite automata and semigroups is much more profound. There are several reasons why this point of view is beneficial. First of all, the semigroup approach to regular languages allows one to use tools from semigroup theory while investigating the properties of these languages. Eilenberg showed that there is a one to one correspondence between certain robust and natural classes of languages and semigroups. This has organized and heightened our understanding of regular languages. Furthermore, in certain computational models, the complexity of a regular language can be parametrized by the complexity of the corresponding semigroup and so this provides us alternate avenues to analyze the complexity of regular languages. Often the combinatorial descriptions of regular languages suffice to obtain upper bounds on their complexity. The algebraic point of view proves to be useful when proving hardness results. Communication complexity is one of the computational models where this is the case.

\section{Outline}

In this thesis, we study the non-deterministic communication complexity of regular languages. The ultimate goal is to find functions $f_1(n), f_2(n), ... , f_k(n)$ such that each regular language has $\Theta (f_i(n))$ non-deterministic communication complexity for some $i \in \{1,2,...,k\}$. Furthermore, we would like a characterization of the languages with $\Theta(f_i(n))$ complexity for all $i \in \{1,2,...,k\}$. In \cite{TT03}, this goal was reached for the following communication models: deterministic, simultaneous, probabilistic, simultaneous probabilistic and Mod$_p$-counting. Obtaining a similar result for the non-deterministic model requires a refinement of the techniques used in \cite{TT03}.

The study of the non-deterministic communication complexity of regular languages from an algebraic point of view is important for several reasons. We can summarize it by stating that it increases our understanding of regular languages and non-deterministic communication complexity. 

From regular languages perspective, our results yield sufficient algebraic conditions for not being in a certain class of languages. This is an interesting result within algebraic automata theory. Furthermore, given the fact that communication complexity has many ties with other computational models, understanding the communication complexity of regular languages helps us understand the power of regular languages in different computational models and where they stand within the complexity theory frame.

From a communication complexity perspective, there are several interesting consequences. In \cite{TT03}, it was shown that in the regular languages setting, $\Theta(\log \log n)$ probabilistic communication complexity coincides with $\Theta(\log n)$ simultaneous communication complexity. Results about the non-deterministic communication complexity leads to further such correspondences which allows us the compare different communication models within the regular languages framework. Even though regular languages are ``simple" with respect to Turing Machines for example, they provide a non-trivial case-study of non-deterministic communication complexity since there are both ``hard" and ``easy" regular languages with respect to this model. Therefore, a complete characterization of regular languages in this model is likely to force one to develop new lower bound techniques and study functions (for example promise functions) other than the commonplace ones which have been intensively studied.

Through the notion of \emph{programs over monoids} (\cite{Bar86}), a connection between algebraic automata theory and circuit complexity has been formed. For example algebraic characterizations of some of the most studied circuit classes $AC^0$, $ACC^0$ and $NC^1$ have been obtained (\cite{BT87}). The connection between communication complexity and circuit complexity is well known. Currently, techniques from communication complexity provide one of the most powerful tools for proving circuit lower bounds (\cite{Gro92},\cite{HG91},\cite{Nis93},\cite{RM97}). Algebraic characterization of regular languages with respect to communication complexity completes a full circle and further strengthens our understanding of the three fields.

\begin{center}
\includegraphics[scale=0.8]{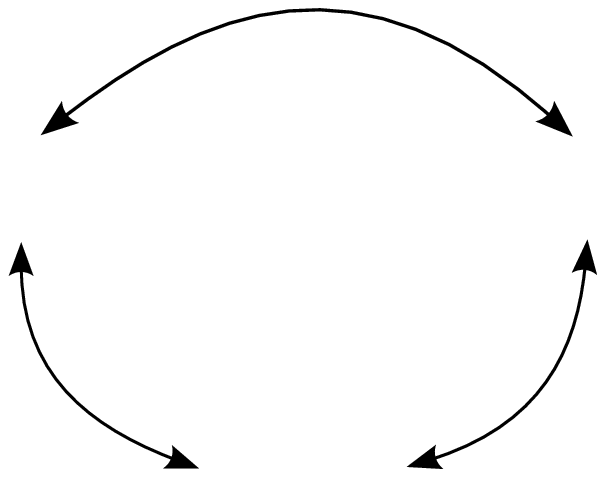}
\rput{0}(-4,1.5){Circuits}
\rput{0}(-7,3.8){Communication Complexity}
\rput{0}(-1,3.8){Algebraic Automata Theory}
\end{center}

The breakdown of the thesis is as follows. In Chapter 2, we give an introduction to communication complexity and present the fundamental techniques in this field. Chapter 3 is devoted to the basics of algebraic automata theory. The main purpose of these two chapters is to deliver the background material needed for Chapter 4. In Chapter 4, we present the results obtained about the non-deterministic communication complexity of regular languages. Finally we conclude in Chapter 5. 

%% file: chap2.tex
\chapter{Communication Complexity}\label{chapter2}

In this chapter, we present the notion of communication complexity as introduced by Yao in \cite{Yao79}. We start in Section 2.1 with the deterministic model in which we look at the fundamental concepts. In Section 2.2, we move to the non-deterministic model which is the model of interest for this work. In Section 2.3, we briefly mention other popular communication models. In Section 2.4, we introduce the notion of a reduction which plays a key role in our arguments in Chapter 4. We also define communication complexity classes and see a beautiful analogy between these classes and Turing Machine classes. Finally we summarize this chapter in Section 2.5.

We refer the reader to the much celebrated book by Kushilevitz and Nisan \cite{KN97} for an in depth survey of the subject. One can also find and excellent introduction in the lecture notes by Ran Raz \cite{Raz04}. We mostly use the notation used in \cite{KN97}.

\section{Deterministic Model}

\subsection{Definition}

In the two-party communication complexity model, we have two players (usually referred to as Alice and Bob) and a function $f: X \times Y \to Z$. Alice is given $x \in X$ and Bob is given $y \in Y$. Both know the function $f$ and their goal is to \emph{collaboratively} compute $f(x,y)$ i.e. they both want to know the value $f(x,y)$. In order to do this they have to communicate (for most functions) since neither of them see the whole input. We are only interested in the number of bits that they need to communicate to compute $f(x,y)$. Thus the complexity of their individual computations are irrelevant and we assume that both Alice and Bob have unlimited computational power.

The communication of Alice and Bob is carried out according to a \index{protocol} protocol $\mathcal{P}$ that both players have agreed upon beforehand. The protocol $\mathcal{P}$ specifies in each step the value of the next bit communicated as a function of the input of the player who sends it and the sequence of previously communicated bits by the two players. The protocol also determines who sends the next bit as a function of the bits communicated thus far.

More formally, a protocol is a 5-tuple of functions $(c_A,c_B,n,f_A,f_B)$ such that:
\begin{itemize}
\item At each step of the communication, $c_A$ takes as input the communication history thus far and the input for Alice and returns the bit that Alice will communicate (similarly for $c_B$ and Bob).
\item $n$ takes as input the communication history thus far and decides whether the communication is over or not. If not, it decides who speaks next.
\item After the communication is over, $f_A$ takes as input the communication history and the input for Alice and returns one bit (similarly for $f_B$ and Bob). This bit is the output of the protocol and the values of $f_A$ and $f_B$ should be the same.
\end{itemize}

Unless stated otherwise, the functions we consider in this chapter are of the form $f:\{0,1\}^n \times \{0,1\}^n \to \{0,1\}$. Since the output of the function is just one bit, we can assume that the last bit communicated is this value.


Let $\mathcal{P}(x,y)$ be the output of the protocol $\mathcal{P}$, i.e. the last bit communicated. Then we say that $\mathcal{P}$ is a protocol for $f$ if for all $x,y \in \{0,1\}^n$, $\mathcal{P}(x,y) = f(x,y)$. The cost of $\mathcal{P}$ is defined as
\[
cost(\mathcal{P}) := \max_{(x,y)\in X \times Y} \textrm{number of bits communicated for $(x,y)$}.
\]
The \index{communication complexity ! deterministic} \emph{deterministic communication complexity} of a function $f$, denoted as $D(f)$, is defined as
\[
D(f) := \min_{\mathcal{P} \textrm{ protocol for $f$}} cost(\mathcal{P})
\]

\begin{example}
Define the \index{EQUALITY} EQUALITY function as
\[
EQ(x,y) := \left\{
          \begin{array}{ll}
          1 & \quad \textrm{if $x = y$,}  \\
          0 & \quad \textrm{otherwise.} \\
          \end{array}
          \right.
\]
An obvious upper bound for $D(EQ)$ is $n+1$ since one of the players, say Alice, can just send all her bits to Bob and Bob can simply compare the string he has with the string Alice has sent. If the strings are equal he can send 1 to Alice and otherwise he can send 0. In fact, this protocol gives an upper bound for any boolean function. Once one of the players knows all the input, s/he can compute the value of the function and send this value to the other player. The number of bits communicated is $n+1$.
\end{example}

Intuitively, one expects that $D(EQ) = n+1$, i.e. $n+1$ is also a lower bound for $D(EQ)$. Although this intuition is correct, how can one rigorously prove this lower bound?

\subsection{Lower Bound Techniques}

As in any other computational model, proving tight lower bounds for the complexity of explicit functions in the communication model is usually a non-trivial task. Nevertheless, there are a number of effective techniques one can use to accomplish this. Now we explore three of these methods: the disjoint cover method, the rectangle size method and the fooling set method.

\begin{figure}
\[
\begin{psmatrix}%
[rowsep=0pt, colsep=3.5pt]
 & 000 & 001 & 010 & 010 & 100 & 101 & 110 & 111 \\
000 & 1 & 0 & 1 & 1 & 0 & 0 & 0 & 1 \\
001 & 0 & 0 & 0 & 0 & 1 & 0 & 1 & 0 \\
010 & 0 & 0 & 0 & 0 & 0 & 0 & 1 & 1 \\
011 & 1 & 0 & 0 & 1 & 1 & 0 & 0 & 1 \\
100 & 1 & 0 & 0 & 0 & 0 & 0 & 1 & 0 \\
101 & 0 & 1 & 1 & 0 & 0 & 0 & 0 & 1 \\
110 & 0 & 0 & 0 & 1 & 0 & 1 & 1 & 1 \\
111 & 0 & 0 & 1 & 1 & 1 & 0 & 0 & 1
\end{psmatrix}
\pspolygon(-5.9,5)(0.1,5)(0.1,-0.4)(-5.9,-0.4)
\pspolygon[linestyle=dashed]%
(-5.1,3)(-5.1,4.3)(-2.9,4.3)(-2.9,3)
\pspolygon[linestyle=dashed]%
(-2.1,3)(-2.1,4.3)(-1.4,4.3)(-1.4,3)
\pspolygon[linestyle=dashed]%
(-5.1,2.35)(-5.1,1.7)(-2.9,1.7)(-2.9,2.35)
\pspolygon[linestyle=dashed]%
(-2.1,2.35)(-2.1,1.7)(-1.4,1.7)(-1.4,2.35)
\]
\label{inputmatrix}
\caption{Example of an input matrix and a 0-monochromatic rectangle.}
\end{figure}

For a function $f$, define the \index{input matrix} \emph{input matrix} by $A^f_{xy} = f(x,y)$ where the rows are indexed by $x \in X$ and the columns are indexed by $y \in Y$. We say that $R$ is a \index{rectangle} \emph{rectangle} if $R = S \times T$ for some $S \subseteq X$ and $T \subseteq Y$. This is equivalent to saying that $(x_1,y_1) \in R$ and $(x_2,y_2) \in R$ together imply $(x_1,y_2) \in R$. We say that $R$ is \index{monochromatic rectangle} \emph{monochromatic} with respect to $f$ if for some $z \in Z$ we have $A^f_{xy} = z$ for all $(x,y) \in R$ (see Figure \ref{inputmatrix}).

Given $f$, let $\mathcal{P}$ be a protocol for $f$. For simplicity let us assume that the players send bits alternately. Also without loss of generality we can assume Alice (who has the input $x$) sends the first bit. Thus at step 1, the protocol partitions $X \times Y = (X_0 \times Y) \cup (X_1 \times Y)$ such that
\[ \forall x \in X_0, \textrm{ Alice sends 0}, \]
\[ \forall x \in X_1, \textrm{ Alice sends 1}. \]
At the second step it is Bob's turn to send a bit so the protocol partitions both $X_0 \times Y = (X_0 \times Y_{00}) \cup (X_0 \times Y_{01})$ and $X_1 \times Y = (X_1 \times Y_{10}) \cup (X_1 \times Y_{11})$. Here, if the first communicated bit was a 0, then
\[ \forall y \in Y_{00}, \textrm{ Bob sends 0}, \]
\[ \forall y \in Y_{01}, \textrm{ Bob sends 1}, \]
and if the first communicated bit was a 1, then
\[ \forall y \in Y_{10}, \textrm{ Bob sends 0}, \]
\[ \forall y \in Y_{11}, \textrm{ Bob sends 1}. \]
In general, if it is Alice's turn to speak and the bits communicated thus far are $b_1,b_2,...b_k$, then Alice partitions $X_{b_1,...,b_{k-1}} \times Y_{b_1,...,b_k}$ into $X_{b_1,...,b_k,0} \times Y_{b_1,...,b_k}$ and $X_{b_1,...,b_k,1} \times Y_{b_1,...,b_k}$. A \index{protocol partitioning tree} protocol partitioning tree nicely illustrates what happens (see Figure \ref{protocoltree}).

\begin{figure}
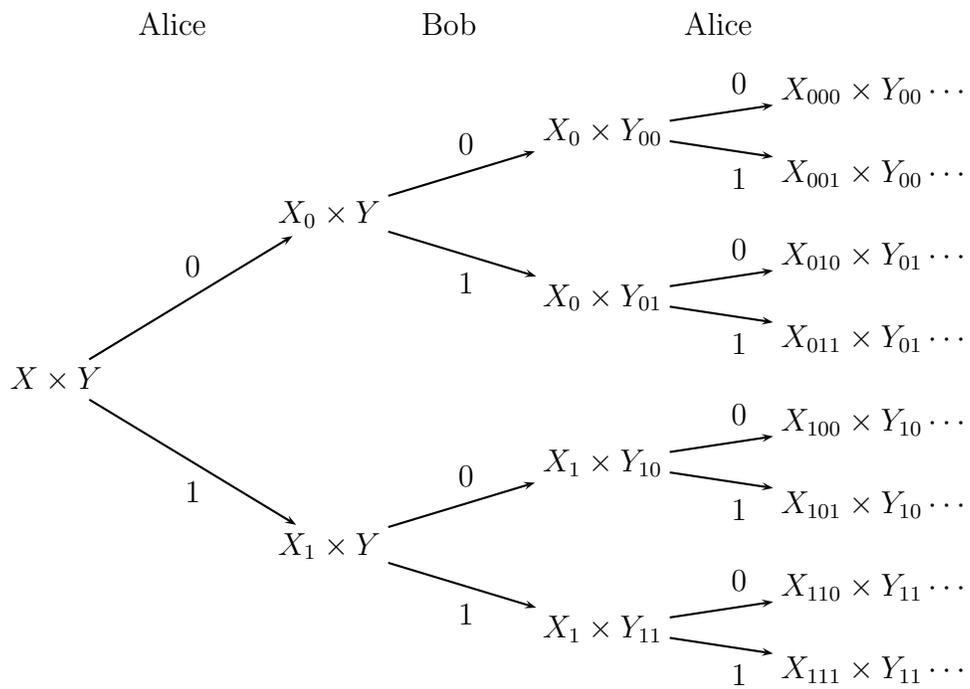

\pstree[treemode=R,levelsep=20ex,edge=none]{\TR{\hspace{1.7cm}Alice}}{
    \pstree{\TR{\hspace{1.8cm}Bob}}{
        \TR{\hspace{1.7cm}Alice}
    }
}
\\\\
\psset{arrows=->}
\pstree[treemode=R,levelsep=20ex,nodesep=3pt]{\TR{$X \times Y$}}{
	\pstree{\TR{$X_0 \times Y$}\taput{0}}{
		\pstree{\TR{$X_0 \times Y_{00}$}\taput{0}}{
			\TR{$X_{000} \times Y_{00}\cdots$}\taput{0}
			\TR{$X_{001} \times Y_{00}\cdots$}\tbput{1}
		}
		\pstree{\TR{$X_0 \times Y_{01}$}\tbput{1}}{
			\TR{$X_{010} \times Y_{01}\cdots$}\taput{0}
			\TR{$X_{011} \times Y_{01}\cdots$}\tbput{1}
		}
	}
    \pstree{\TR{$X_1 \times Y$}\tbput{1}}{
		\pstree{\TR{$X_1 \times Y_{10}$}\taput{0}}{
			\TR{$X_{100} \times Y_{10}\cdots$}\taput{0}
			\TR{$X_{101} \times Y_{10}\cdots$}\tbput{1}
		}
		\pstree{\TR{$X_1 \times Y_{11}$}\tbput{1}}{
			\TR{$X_{110} \times Y_{11}\cdots$}\taput{0}
			\TR{$X_{111} \times Y_{11}\cdots$}\tbput{1}
		}
	}
}
\caption{Protocol partitioning tree.}
\label{protocoltree}
\end{figure}

Observe that each node in the protocol partitioning tree is a rectangle and two nodes intersect if and only if one is the ancestor of the other. In particular, the leaves of the tree are disjoint rectangles. The same bits are communicated for all the inputs in a leaf so $\mathcal{P}(x,y)$ is the same for all these inputs, i.e. the leaves are monochromatic. The height of the tree is equal to $cost(\mathcal{P})$. Thus we have proved the following lemma which is a key combinatorial property of a protocol.

\begin{lemma}
A protocol $\mathcal{P}$ for $f$ with $cost(\mathcal{P}) = c$ partitions the input matrix $A^f$ into at most $2^c$ monochromatic rectangles.
\end{lemma}

A \index{monochromatic disjoint cover} \emph{monochromatic disjoint cover} is a partition of a matrix into disjoint monochromatic rectangles. We denote by $C^D(f)$, the minimum number of rectangles in any monochromatic disjoint cover of $A^f$. With this definition and the previous lemma at hand, we can present the first lower bound technique.

\index{disjoint cover method}
\begin{corollary}[Disjoint Cover Method]\label{disjointcovermethod}
For a function $f$ we have $D(f) \geq \log_2 C^D(f)$.
\end{corollary}

With this tool it is now easy to show a linear lower bound for $D(EQ)$. Observe that the input matrix for EQUALITY is a $2^n$ by $2^n$ identity matrix. No 1-monochromatic rectangle can contain more than one 1. Thus any monochromatic disjoint cover has $2^n$ 1-monochromatic rectangles and at least one 0-monochromatic rectangle. So $D(EQ) \geq \lceil \log_2 (2^n + 1) \rceil = n + 1$ as predicted.

Although every protocol for a function induces a monochromatic disjoint cover of the input matrix, simple examples show that the converse is not true. So if some of the monochromatic disjoint covers do not correspond to any protocol, how good can the disjoint cover method be? The next theorem states that the gap is not very large.

\begin{theorem}\label{halstenberg}
For a function $f$, we have $D(f) = O( \log_2^2 C^D(f) )$.
\end{theorem}
\begin{proof}
For any function $f$, we present a protocol for it with complexity $O( \log_2^2 C^D(f) )$. The protocol consists of at most $\log_2 C^D(f)$ rounds and in each round at most $\log_2 C^D(f) + O(1)$ bits are communicated. The basic idea is as follows: \\
Alice and Bob agree upon an optimal disjoint monochromatic cover beforehand. They try to figure out whether $(x,y)$ lies in a 0-monochromatic rectangle or a 1-monochromatic rectangle. The protocol proceeds in rounds. If $f(x,y)=1$ then in each round they successfully eliminate at least half of the 0-monochromatic rectangles. At the end, all 0-monochromatic rectangles are eliminated and they conclude $f(x,y)=1$. If on the other hand $f(x,y)=0$, then in one of the rounds they are not able to eliminate at least half of the 0-monochromatic rectangles. At this point they conclude $f(x,y)=0$.\\
Before giving the details of a round, we make two crucial observations. The first observation implies the second one. The correctness of the protocol follows from the second observation. \\
\underline{Observation 1:} Suppose $R_0 = S_0 \times T_0$ is a 0-monochromatic rectangle and $R_1 = S_1 \times T_1$ is a 1-monochromatic rectangle. Then either $R_0$ and $R_1$ are disjoint in rows or they are disjoint in columns, i.e. either $S_0$ and $S_1$ are disjoint or $T_0$ and $T_1$ are disjoint.\\
\underline{Observation 2:} Let $\mathcal{C}$ be any collection of 0-monochromatic rectangles and $R_1$ any 1-monochromatic rectangle. Then either
\begin{itemize}
\item[-] $R_1$ intersects with at most half of the rectangles in $\mathcal{C}$ in rows or
\item[-] $R_1$ intersects with at most half of the rectangles in $\mathcal{C}$ in columns.
\end{itemize}
Otherwise there is at least one rectangle $R_0$ in $\mathcal{C}$ such that $R_0$ and $R_1$ intersect both in rows and columns. This contradicts the first observation.\\
Now we can describe how a round is carried out. Initially $\mathcal{C}$ contains all the 0-monochromatic rectangles.
\begin{itemize}
\item[A.] If $\mathcal{C} = \emptyset$ then Alice communicates to Bob that $f(x,y) = 1$ and the protocol ends. Otherwise, Alice tries to find a 1-monochromatic rectangle $R_1 = S_1 \times T_1$ such that $x \in S_1$ and $R_1$ intersects with at most half of the rectangles in $\mathcal{C}$ in rows. If such a rectangle exists, then Alice sends its name ($\log_2 C^D(f)$ bits) to Bob and they both update $\mathcal{C}$ so it contains all the rectangles that intersect with $R_1$ in rows (the other rectangles cannot contain $(x,y)$). At this point the round is over since they successfully eliminated at least half of the rectangles in $\mathcal{C}$. If Alice is unable find such a rectangle then she communicates this to Bob.
\item[B.] At this point we know Alice could not find a 1-monochromatic rectangle to end the round so Bob tries to end the round by finding a 1-monochromatic rectangle $R_1 = S_1 \times T_1$ such that $y \in T_1$ and $R_1$ intersects with at most half of the rectangles in $\mathcal{C}$ in columns. If he finds such a rectangle, he communicates its name to Alice and they both update $\mathcal{C}$ so it contains all the rectangles that intersect with $R_1$ in columns. After this point the round is over. If he cannot find such a rectangle this means both Alice and Bob failed and therefore he communicates to Alice that $f(x,y) = 0$ because by the second observation, he knows that there is no 1-monochromatic rectangle containing $(x,y)$.
\end{itemize}
\end{proof}

In most cases it is hard to exactly determine $C^D(f)$. So the natural next step is to find lower bounds on $C^D(f)$ which in turn gives lower bounds on $D(f)$. (This is actually what we did for the EQUALITY function.)

An obvious way of bounding (from below) the number of monochromatic rectangles needed in a monochromatic disjoint cover is to bound (from above) the size of every monochromatic rectangle. In other words, if every monochromatic rectangle in the input matrix has size less than or equal to $s$, then we need at least $2^{2n} / s$ monochromatic rectangles in a monochromatic disjoint cover of the matrix. Here `size' refers to the number of pairs $(x,y)$ in the rectangle and we can interpret this as a measure $\mu$. The above actually generalizes to any kind of measure.

\index{rectangle size method}
\begin{proposition}[Rectangle Size Method]\label{rectanglesize}
Let $\mu$ be a measure defined on the space $X \times Y$. If all monochromatic rectangles $R$ (with respect to $f$) are such that $\mu(R) \leq s$, then $D(f) \geq \log_2 ( \mu(X \times Y) / s )$.\\
In particular, if $\mu$ is a probability and every monochromatic rectangle $R$ satisfies $\mu(R) \leq \epsilon$, then $D(f) \geq \log_2 1 / \epsilon$.
\end{proposition}

\begin{example}
Let us see an application of the rectangle size method by proving a linear lower bound for the communication complexity of the \index{DISJOINTNESS} DISJOINTNESS function. We define DISJOINTNESS as
\[
DISJ(x,y) := \left\{
          \begin{array}{ll}
          1 & \quad \textrm{if $x \cap y = \emptyset$,}  \\
          0 & \quad \textrm{otherwise.} \\
          \end{array}
          \right.
\]
where $x$ and $y$ are viewed as subsets of $[n]$ ($x_i = 1$ if $x$ contains the element $i \in [n]$). We claim that any 1-monochromatic rectangle $R = S \times T$ has size at most $2^n$. It is easy to show that the number of $(x,y)$'s such that $x \cap y = \emptyset$ is $\sum_{j=0}^n {n \choose j} 2^{n-j} = 3^n$. Now if for all $x$ and $y$ that intersect we set $\mu(x,y) = 0$ and for all $x$ and $y$ that are disjoint we set $\mu(x,y) = 1$ then $\mu(X \times Y) = 3^n$ and the above claim together with Proposition \ref{rectanglesize} imply $D(DISJ) = \Omega(n)$.\\
{\it Proof of claim:} Suppose $|S| = k$ and $|\cup_{x \in S} x| = c$. Then clearly $k \leq 2^c$. Also $|T| \leq 2^{n-c}$ since every set in $T$ must be disjoint from every set in $S$. Thus the size of the rectangle is $|S| \cdot |T| \leq k 2^{n-c} \leq 2^c 2^{n-c} = 2^n$.
\end{example}

The last lower bound technique we look at in this section is the well-known fooling set technique. It is a direct consequence of the disjoint cover method and in fact it is a special case of Proposition \ref{rectanglesize}. First we make the formal definition of a fooling set.

\begin{definition}
A set $F \subseteq X \times Y$ is a \index{fooling set} \emph{fooling set} for $f$ if the following conditions are satisfied.
\begin{enumerate}
\item For all $(x,y) \in F$, $f(x,y) = z$ for some $z \in Z$.
\item For all distinct $(x_1,y_1),(x_2,y_2) \in F$ either $f(x_1,y_2)\neq z$ or $f(x_2,y_1)\neq z$.
\end{enumerate}
\end{definition}

By the definition of a fooling set, no two elements in $F$ can be in the same monochromatic rectangle. Therefore there must be at least $|F|$ many monochromatic rectangles in any monochromatic disjoint cover of the input matrix. So by Corollary \ref{disjointcovermethod} we get the following fact.

\index{fooling set method}
\begin{lemma}[Fooling Set Method]\label{foolingset}
If $F$ is a fooling set for $f$ then $D(f) \geq \log_2 |F|$.
\end{lemma}

To see that the fooling set method is indeed a special case of Proposition \ref{rectanglesize}, for a fooling set $F$, let $\mu(x,y) = c > 0$ for every $(x,y) \in F$ and for every $(x,y) \notin F$ set $\mu(x,y) = 0$. Then any monochromatic rectangle $R$ satisfies $\mu(R) \leq c$ and therefore \[ D(f) \geq \log_2 ( \mu(X \times Y) / c ) = \log_2 ( c|F| / c ) = \log_2 |F|. \]

\begin{example}
Define the \index{LESS-THAN} LESS-THAN function as
\[
LT(x,y) := \left\{
          \begin{array}{ll}
          1 & \quad \textrm{if $x \leq y$,}  \\
          0 & \quad \textrm{otherwise.} \\
          \end{array}
          \right.
\]
where $x$ and $y$ are viewed as binary numbers. We can show that $LT$ has linear deterministic communication complexity by the fooling set technique. Let $F = \{(x,x) : x \in \{0,1\}^n\}$. It is easy to see that $F$ is a fooling set. Clearly $|F|=2^n$ and this proves our claim. In fact $F$ is also a fooling set for the EQUALITY function.
\end{example}

From our discussion in this section, it is clear that we can exploit the nice combinatorial structure of protocols to prove tight lower bounds for explicit functions. In the next section, we see that most of the techniques seen in this section can be applied to the non-deterministic model as well.

\section{Non-Deterministic Model}

\subsection{Definition}

The definition of the non-deterministic communication model is analogous to the non-deterministic model in the Turing Machine world. There are several ways of defining non-determinism, all of which are equivalent. Here we will present the one that best suits our needs.

Intuitively, non-determinism can be viewed as a certificate verification process\footnote{Equivalently one can view it as a communication game in which the players are allowed to take non-deterministic steps.}: A third player (referred to as God) gives a proof (bit string) that $f(x,y)=z$ to both Alice and Bob. If indeed $f(x,y)=z$, then Alice and Bob must be able to convince themselves that this is the case by communicating with each other. If on the other hand $f(x,y)\neq z$, then the verification process should fail and Alice and Bob should be able to conclude that the proof was wrong. We consider the bits sent by God as a part of the communicated bits.

More formally, in the non-deterministic setting, Alice and Bob communicate according to a non-deterministic protocol $\mathcal{P}^z$. This protocol differs from the deterministic one as follows. $\mathcal{P}^z$ takes three inputs, $x,y$ and $s$, where $x$ and $y$ are perceived as the inputs for Alice and Bob respectively, and $s$ is some bit string which we think of as the ``proof string". $\mathcal{P}^z$ specifies in each step the value of the next bit communicated as a function of the input of the player who sends it, the sequence of previously communicated bits as well as $s$. It also determines who will send the next bit as a function of the communicated bits thus far. So it differs from a deterministic protocol because what a player sends also depends on the string $s$. We will denote the output of the protocol by $\mathcal{P}^z(x,y,s)$.

We say that $\mathcal{P}^z$ is a non-deterministic protocol for $f$ if for all $(x,y)$ such that $f(x,y) = z$, there exists a string $s$ such that $\mathcal{P}^z(x,y,s) = z$, and for all $(x,y)$ such that $f(x,y) \neq z$ we have $\mathcal{P}^z(x,y,s) \neq z$ for any $s$.

The the cost of $\mathcal{P}^z$ is defined as
\[ cost(\mathcal{P}^z) := \max_{\begin{subarray}{c}
(x,y): \\
f(x,y)=z
\end{subarray}} \min_{\begin{subarray}{c}
s: \\
\mathcal{P}^z(x,y,s)=z
\end{subarray}} |s| + \textrm{ no of communicated bits for $(x,y,s)$.}
\]

We define the \index{communication complexity ! non-deterministic} \emph{non-deterministic communication complexity} of $f$ as
\[ N^1(f) := \min_{\mathcal{P}^1 \textrm{ non-deterministic protocol for $f$}} cost(\mathcal{P}^1) 
\]
The co-non-deterministic communication complexity of $f$ is defined similarly and is denoted by $N^0(f)$.

\begin{example}\label{DISJnondeterministic}
Let us show $N^0(DISJ) = O(\log_2 n)$ by exhibiting a proof and a verification protocol. God can prove $DISJ(x,y) = 0$ by telling Alice and Bob the index $i$ in which $x$ and $y$ intersect. This proof is $O(\log_2 n)$ bits long and Alice and Bob can convince themselves that the proof is correct by exchanging the bits $x_i$ and $y_i$. If $DISJ(x,y)=1$, then given any index as a proof, Alice and Bob can detect that the proof is wrong. (Any other kind of proof is considered as a wrong proof.)
\end{example}

Unlike the deterministic communication complexity, we can get an exact characterization of non-deterministic communication complexity in terms of monochromatic rectangles. We denote by $C^z(f)$ the minimum number of $z$-monochromatic rectangles in any \index{monochromatic cover} monochromatic cover of the $z$-inputs of $f$ (observe that here we dropped the word ``disjoint" since we allow the rectangles to intersect). This quantity exactly determines $N^z(f)$.

\begin{figure}
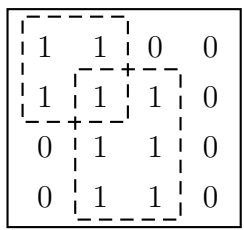

\[
\begin{psmatrix}%
[rowsep=0pt, colsep=15pt]
1 & 1 & 0 & 0 \\
1 & 1 & 1 & 0 \\
0 & 1 & 1 & 0 \\
0 & 1 & 1 & 0
\end{psmatrix}
\pspolygon(-2.8,2.5)(0.3,2.5)(0.3,-0.4)(-2.8,-0.4)
\pspolygon[linestyle=dashed]%
(-2.6,2.4)(-1.2,2.4)(-1.2,1)(-2.6,1)
\pspolygon[linestyle=dashed]%
(-1.9,1.7)(-0.5,1.7)(-0.5,-0.3)(-1.9,-0.3)
\]
\caption{An example of a monochromatic cover of the 1-inputs.}
\end{figure}

\begin{proposition}
$\log_2 C^z(f) \leq N^z(f) \leq \log_2 C^z(f) + 2$.
\end{proposition}
\begin{proof}
\hfill
\begin{itemize}
\item $\log_2 C^z(f) \leq N^z(f)$\\
We have seen in the deterministic case that a certain communication pattern corresponds to a certain monochromatic rectangle. This situation is not much different in the non-deterministic model. In this case, what Alice and Bob send in each step also depends on the proof bits. So for every fixed proof string, there corresponds a protocol partitioning tree as in Figure \ref{protocoltree}.

Now observe that for this particular proof, every communication pattern that convinces Alice and Bob leads to a $z$-monochromatic rectangle. (Other communication patterns may not lead to a monochromatic rectangle since the proof we fixed may not be a proof for all $(x,y)$ with $f(x,y)=z$.) So each convincing communication pattern (including the proof) corresponds to a $z$-monochromatic rectangle. Since for every $(x,y)$ such that $f(x,y)=z$ there must be a proof that convinces Alice and Bob, all the convincing communication patterns together correspond to a covering of the $z$-inputs. Here the rectangles are allowed to intersect since for some $(x,y)$ with $f(x,y)=z$, there might be more than one proof that leads Alice and Bob to be convinced. There are at most $2^{N^z(f)}$ communication patterns and therefore $C^z(f) \leq 2^{N^z(f)}$.

\item $N^z(f) \leq \log_2 C^z(f) + 2$\\
Fix any optimal monochromatic cover of the $z$-inputs. If God sends Alice and Bob the name of a monochromatic rectangle $R = S \times T$ that $(x,y)$ lies in, then Alice can check that $x \in S$ and if so, she can send 1 to Bob. Bob can similarly check if $y \in T$ and send 1 to Alice if this is the case.
\end{itemize}
\end{proof}

\subsection{Power of Non-Determinism}

A natural question that arises in this context is: how much power does non-determinism give? Non-determinism in the finite automaton computational model does not give extra power with respect to the class of languages recognized. In the Turing Machine model, it is not known whether non-determinism provides significantly more power. In the communication complexity model we can answer this question and prove that non-determinism is strictly more powerful. First we observe that the gap between determinism and non-determinism cannot be more than exponential.

\begin{proposition}
For any $z \in \{0,1\}$, $D(f) \leq C^z(f) + 1$.
\end{proposition}
\begin{proof}
Alice and Bob agree on an optimal cover of the $z$-inputs. Alice communicates to Bob the $z$-monochromatic rectangles that $x$ lies in (this requires $C^z(f)$ bits of communication). Bob, with this information, can determine if there is a $z$-monochromatic rectangle that $(x,y)$ lies in and send the answer to Alice.
\end{proof}

The above is actually tight. For example the EQUALITY function satisfies $D(EQ) = n+1$ and $N^0(EQ) \leq \log_2 n + 2$ (similar protocol to the one in Example \ref{DISJnondeterministic}).

Can it be the case that both $N^0(f)$ and $N^1(f)$ are exponentially smaller than $D(f)$? The answer to this question is given by the next theorem.

\begin{theorem}\label{detvsnondet}
For every function $f:X \times Y \to \{0,1\}$,
\[ D(f) = O( N^0(f) N^1(f) ). \]
\end{theorem}

The proof of this theorem is the same as the proof of Theorem \ref{halstenberg}. It was shown in \cite{Fur87} that this bound is tight.

Observe that there are two reasons why non-determinism is more powerful than determinism:
\begin{enumerate}
\item non-determinism is one sided in the sense that we only need to cover the $z$-inputs,
\item the $z$-monochromatic rectangles in the cover are allowed to overlap.
\end{enumerate}
From our discussion above it should be clear that the ultimate power comes from the first point. In the EQUALITY example we see that it is ``easy" to cover the 0-inputs in the sense that we do not need exponentially many 0-monochromatic rectangles to cover the 0-inputs. The hardness lies in covering the 1-inputs. The exponential gap is a product of this fact. The power of a cover against a disjoint cover is only quadratic as implied by Theorem \ref{detvsnondet}.

\subsection{Lower Bound Techniques}

In the deterministic model, we saw the rectangle size method as a lower bound technique. It is clear that the same approach gives a lower bound for the non-deterministic communication complexity as well. If every $z$-monochromatic rectangle has size less than or equal to $s$ and there are $k$ $z$-inputs, then we need at least $k/s$ many rectangles to cover these inputs. The non-deterministic version of Proposition \ref{rectanglesize} is as follows.

\begin{proposition}
Let $K \subseteq X \times Y$ be the set of all $z$-inputs and let $\mu$ be a measure defined on the space $K$. If all $z$-monochromatic rectangles $R$ satisfy $\mu(R) \leq s$, then $N^z(f) \geq \log_2(\mu(K)/s)$.
\end{proposition}

It can be shown that the rectangle size method in the non-deterministic case is almost tight. Suppose we choose the best possible measure $\mu$ (i.e. the one that gives the best bound) and the maximum size (with respect to $\mu$) of a $z$-monochromatic rectangle is $s$. Then we have:

\begin{theorem}[see \cite{KN97}]
$N^z(f) \leq \log_2(\mu(K)/s) + \log_2 n + O(1)$.
\end{theorem}

There are examples that show that we cannot do better than this.

The fact that we can use the rectangle size method here implies that we can also use the fooling set method. However, as the next proposition shows, the quality of the fooling set method is questionable.

\begin{proposition}[see \cite{KN97}]
Almost all functions $f: \{0,1\}^n \times \{0,1\}^n \to \{0,1\}$ satisfy $N^1(f) = \Omega(n)$ but the size of their largest fooling set is $O(n)$.
\end{proposition}

We finish off this section by looking at the non-deterministic communication complexity of the \index{PROMISE-DISJOINTNESS} PROMISE-DISJOINTNESS function. This function is defined the same way as the DISJOINTNESS function but the input space is different: it is the union of the following two sets $A$ and $B$.
\[ A := \{(x,y) \in \{0,1\}^n \times \{0,1\}^n : x \cap y = \emptyset \}. \]
\[ B := \{(x,y) \in \{0,1\}^n \times \{0,1\}^n : |x \cap y| = 1 \}. \]
In other words Alice and Bob are promised that if they get an input that intersects, then the size of the intersection is exactly 1.

Now we show that the PROMISE-DISJOINTNESS ($PDISJ$) function has linear non-deterministic complexity. This fact is used in Chapter \ref{complexityofregularlanguages} to prove linear lower bounds for the complexity of certain regular languages. To show the linear lower bound, we use a result that implies a linear lower bound on the randomized communication complexity of the DISJOINTNESS function. Before we can state this result, we first need to define two measures on $\{0,1\}^n \times \{0,1\}^n$.
\[
\mu_A(x,y) := \left\{
          \begin{array}{ll}
          \frac{1}{|A|} & \quad \textrm{if $(x,y) \in A$,}  \\
          0 & \quad \textrm{otherwise.} \\
          \end{array}
          \right.
\]
\[
\mu_B(x,y) := \left\{
          \begin{array}{ll}
          \frac{1}{|B|} & \quad \textrm{if $(x,y) \in B$,}  \\
          0 & \quad \textrm{otherwise.} \\
          \end{array}
          \right.
\]

\begin{lemma}[see \cite{Raz04}]
For any rectangle $R = S \times T$, if \[\mu_A(R) > 2^{-n/100}\] then
\[ \mu_B(R) > \frac{1}{100} \mu_A(R) \]
(for $n$ large enough).
\end{lemma}

In particular, if $\mu_A(R) > 2^{-n/100}$ then $R$ contains elements from both $A$ and $B$. Therefore to cover the inputs in $A$ with 1-monochromatic rectangles\footnote{In this setting 1-monochromatic rectangles can contain any element from $(X \times Y) \backslash B$}, we need exponentially many rectangles. This shows $N^1(PDISJ) = \Omega(n)$.

\section{Other Models}

In this section, we mention some of the most interesting and well-studied communication complexity models.

\index{communication complexity ! randomized}
\subsubsection*{Randomized Communication Complexity}

In the randomized setting, Alice and Bob both have access to random bit strings that are generated according to some probability distribution. These random strings are private to them and are independent. What Alice and Bob communicate depends on these random strings as well as their input and the previously communicated bits. We say that $\mathcal{P}$ is a protocol for $f$ with $\epsilon$ error if the following holds.
\[ \forall (x,y) \in X \times Y, \quad \textrm{Pr}[ \mathcal{P}(x,y) = f(x,y) ] \geq 1 - \epsilon \]
The cost of $\mathcal{P}$ is defined as the maximum number of bits communicated where the maximum is taken over all possible random strings and all inputs $(x,y)$. The randomized communication complexity of $f$ is
\[ R(f) := \min_{\mathcal{P} \textrm{ protocol for $f$ with error } 1/3} cost(\mathcal{P}).\]

One can also define the one sided error randomized complexity. We say that $\mathcal{P}$ is a protocol for $f$ with one sided $\epsilon$ error if the following holds.
\[ \forall (x,y) \in X \times Y \textrm{ with $f(x,y)=0$}, \quad \textrm{Pr}[ \mathcal{P}(x,y) = 0 ] = 1 \textrm{ and} \]
\[ \forall (x,y) \in X \times Y \textrm{ with $f(x,y)=1$}, \quad \textrm{Pr}[ \mathcal{P}(x,y) = 1 ] \geq 1 - \epsilon. \]
Then the one sided randomized communication complexity of $f$ is
\[ R^1(f) := \min_{\mathcal{P} \textrm{ protocol for $f$ with one sided error } 1/2} cost(\mathcal{P}). \]

There are also variations of the randomized model in which Alice and Bob have access to one public random string. (For a comparison see \cite{New91}.)

\index{communication complexity ! distributional}
\subsubsection*{Distributional Communication Complexity}

In this setting, the definition of the cost of a protocol and the communication complexity of a function are the same as the deterministic model. The difference is that we relax the condition
\[ \forall (x,y) \in X \times Y, \quad \mathcal{P}(x,y) = f(x,y) \]
to
\[ \textrm{Pr}_{\mu}[\mathcal{P}(x,y) = f(x,y)] \geq 1-\epsilon \]
for a given probability distribution $\mu$ on the input space $X \times Y$ and a constant $\epsilon$. The distributional communication complexity of a function is denoted by $D_{\epsilon}^\mu (f)$.

\index{communication complexity ! multiparty}
\subsubsection*{Multiparty Communication Complexity}

A natural way of generalizing the two player model to $k$-players is as follows. $k$-players try to compute a function $f: X_1 \times X_2 \times ... \times X_k \to Z$ where player $i$ gets $x_i \in X_i$ and communication is established by broadcasting (every player receives the communicated bit). Observe that as the number of players increases, the power of the model decreases. 

Another way of generalizing the two party model to $k$ players was proposed in \cite{CFL83}. This model is referred to as ``number on the forehead" model because here each player $i$ sees every input but $x_i$. This can be viewed as each player having their input on their forehead and not being able to see it. The power of this model increases as the number of players increases. In this setting, coming up with lower bounds is considerably harder. However, these lower bounds imply lower bounds in other computational models such as circuits and bounded-width branching programs. This is one of the reasons why this model has attracted more interest than the natural generalization mentioned previously. There are applications in time-space tradeoffs for Turing Machines (\cite{BNS92}), length-width tradeoffs for branching programs (\cite{BNS92}), circuit complexity (\cite{HG91}, \cite{Gro92}, \cite{Nis93}, \cite{Gro98}), proof complexity (\cite{BPS07}) and pseudorandom generators (\cite{BNS92}), to cite only a few.

\section{Communication Complexity Classes}

It is possible to define complexity classes with respect to communication complexity once we settle what it means to be ``easy" or ``tractable". Communication complexity classes were introduced in \cite{BFS86} in which ``tractable" was defined to be $polylog(n)$ complexity. That is, a function is tractable if its complexity is $O(\log^c n)$ for some constant $c$. From this foundation, one can build communication complexity classes analogous to $P, NP, coNP, BPP, RP$ and many more. For example $P^{cc} = \{ f : D(f) = polylog(n)\}$. The correspondence between some of the complexity classes and the complexity measures can be summarized as follows.
\[
\begin{array}{| c || c | c | c | c | c |}
\hline
\textrm{\bf{Complexity class}} & P^{cc} & NP^{cc} & coNP^{cc} & BPP^{cc} & RP^{cc} \\
\hline
\textrm{\bf{Complexity measure}} & D & N^1 & N^0 & R & R^1 \\
\hline
\end{array}
\]

The relationship between these classes are much better known than their Turing Machine counterparts since proving lower bounds for explicit functions is easier in the communication world. We have seen that the function NOT-EQUALITY satisfies $D(NEQ) = n+1$ and $N^1(NEQ) \leq \log_2 n + 1$. This proves $P^{cc} \neq NP^{cc}$. Since $N^0(NEQ) = O(n)$, we have $coNP^{cc} \neq NP^{cc}$. Theorem \ref{detvsnondet} shows that $P^{cc} = NP^{cc} \cap coNP^{cc}$. It can also be shown that $P^{cc} \neq RP^{cc}$ and $NP^{cc} \nsubseteq BPP^{cc}$.
\begin{remark}
It is also possible to define analogs of the polynomial hierarchy.
\end{remark}

Reducibility and completeness are fundamental concepts in the Turing Machine computational model so it is natural to define the communication complexity analogs.

The idea of reduction is as follows. Given two functions $f$ and $g$, $f$ \emph{reduces} to $g$ if Alice and Bob can privately convert their inputs $x$ and $y$ to $x'$ and $y'$ such that $f(x,y) = 1$ if and only if $g(x',y') = 1$. Suppose $f$ reduces to $g$ and that the inputs of length $n$ are converted into inputs of length $t(n)$. Then it is clear that if the communication complexity of $g$ is $O(h(n))$ then the communication complexity of $f$ is $O(h(t(n)))$. If the communication complexity of $f$ is $\Omega(h(n))$ then the communication complexity of $g$ is $\Omega(h(t^{-1}(n)))$.

Reductions of particular interest with respect to the communication complexity classes are those with $t(n) = 2^{\log_2^c n}$ for some constant $c$. The formal definition as given in \cite{BFS86} is as follows.

\begin{definition}
Let $t = 2^{\log_2^c n}$ for some constant $c$. A \index{rectangular reduction} \emph{rectangular reduction} from a function $f:\{0,1\}^n \times \{0,1\}^n \to \{0,1\}$ to a function $g:\{0,1\}^t \times \{0,1\}^t \to \{0,1\}$ is a pair of functions $a:\{0,1\}^n \to \{0,1\}^t$ and $b:\{0,1\}^n \to \{0,1\}^t$ such that $f(x,y) = 1$ if and only if $g(a(x), b(y)) = 1$.
\end{definition}

From this definition it is clear that if there is a rectangular reduction from $f$ to $g$ and $g \in P^{cc}$ then $f \in P^{cc}$. The same is true for $NP^{cc}$ (and in fact for every level of the polynomial hierarchy).

\begin{example}\label{reduction}
For a fixed constant $q > 1$, define the \index{INNER-PRODUCT} INNER-PRODUCT function as follows.
\[
IP_q(x,y) := \left\{
          \begin{array}{ll}
          1 & \quad \textrm{if $\sum_{i=1}^n x_iy_i \equiv 0 \mod q$,}  \\
          0 & \quad \textrm{otherwise.} \\
          \end{array}
          \right.
\]
We exhibit a reduction from $PDISJ$ to $IP_q$ such that an input of length $n$ is converted into an input of length $n+q$. Since $q$ is a constant, this proves that $N^1(IP_q) = \Omega(n)$.\\
Given $x$ and $y$, Alice and Bob each append $q$ 1's at the end of their inputs to obtain $x'$ and $y'$. If $PDISJ(x,y) = 1$ then clearly $IP_q(x',y') = 1$. If on the other hand $PDISJ(x,y)=0$, then we know that $x$ and $y$ intersect only at one position and therefore $x'$ and $y'$ will intersect in $q+1$ positions. This implies $IP_q(x',y') = 0$.
\end{example}

Having established the definition of a reduction, we can define the notion of completeness. For a class of functions $\mathcal{C}$, we say $f \in \mathcal{C}$ is \index{complete} \emph{complete} in $\mathcal{C}$ if there is a rectangular reduction from every function in $\mathcal{C}$ to $f$. In \cite{BFS86}, a complete function is found in every level of the polynomial hierarchy.

\section{Summary}

In this chapter we took a glimpse at the mini-world within complexity theory: communication complexity. The main focus in this area has been proving tight lower bounds for specific functions. We showed three lower bound techniques for the deterministic model. These were the disjoint cover method, rectangle size method and the fooling set method. We introduced the non-deterministic model and saw that the non-deterministic communication complexity of a function was essentially the number of $z$-monochromatic rectangles needed to cover the $z$-inputs. We saw that the rectangle size method, and therefore the fooling set method were also applicable in this setting. We looked at the power of non-determinism and observed that the possible exponential gap between the deterministic and the non-deterministic complexity arose from the fact that non-determinism was one sided. Later we touched on some other communication models: randomized complexity, distributional complexity and multiparty complexity. Finally we defined some of the communication complexity classes, $P^{cc},NP^{cc},coNP^{cc},BPP^{cc},RP^{cc}$, by considering $polylog(n)$ complexity as tractable. Natural definitions of reducibility and completeness were also introduced.

The deterministic and the non-deterministic communication complexities of the functions seen in this chapter are summarized with the following table.
\[
\begin{array}{|c|| c| c| c| c| c| c|}
\hline
 & EQ & NEQ & LT & DISJ & PDISJ & IP_q \\
\hline
\hline
D & \Theta(n) & \Theta(n) & \Theta(n) & \Theta(n) & \Theta(n) & \Theta(n) \\
\hline
N^1 & \Theta(n) & \Theta(\log_2 n) & \Theta(n) & \Theta(n) & \Theta(n) & \Theta(n) \\
\hline
\end{array}
\]

%% file: chap3.tex
\chapter{Algebraic Automata Theory}\label{chapter3}

In this chapter, we introduce the reader to algebraic automata theory by presenting the fundamental concepts in this area. The heart of this theory is viewing a monoid as a language recognizer. Therefore we begin this chapter in Section 3.1 by explaining how a monoid can be viewed as a computational machine. Later we define the syntactic monoid of a language which is analogous to the minimal automaton. Then we define varieties and state the variety theorem which establishes a one to one correspondence between varieties of finite monoids and varieties of regular languages. This conveys the intimate relationship between finite monoids and regular languages. In Section 3.2, we extend the theory to ordered monoids since (as we see in Chapter 4) this provides the proper framework to analyze the non-deterministic communication complexity of regular languages.

We assume that the reader has basic knowledge in automata theory. For more details on the subjects covered in this chapter, see \cite{Pin86} and \cite{Pin97} for the ordered case.


\section{Monoids - Automata - Regular Languages}

\subsection{Monoids: A Computational Model}

Before we can present how a monoid can be viewed as a computational machine, we first need to formally define a monoid and a morphism between monoids. A \index{semigroup} \emph{semigroup} ($S$,$\cdot$) is a set $S$ together with an associative binary operation defined on this set. A \index{monoid} \emph{monoid} ($M$, $\cdot$) is a semigroup that has an identity: $\exists \; 1_M \in M$ which satisfies $1_M \cdot m = m \cdot 1_M = m$ for any $m \in M$. We denote a monoid by its underlying set and write $m_1m_2$ instead of $m_1 \cdot m_2$ when there is no ambiguity about the operation. Observe that a group is just a monoid in which each element has an inverse.

Given two monoids $M$ and $N$, a function $\varphi: M \to N$ is a \index{morphism ! between monoids} \emph{morphism} if $\varphi(1_M) = 1_N$ and if $\varphi$ preserves the operation, i.e. $\varphi(mm') = \varphi(m)\varphi(m')$ for any $m,m' \in M$.

We assume that the monoids we are dealing with are finite, with the exception of the free monoid $\Sigma^*$ which consists of all words (including the empty word $\epsilon$) over the alphabet $\Sigma$, with the underlying operation being concatenation. Observe that any function $\varphi: \Sigma \to M$ extends uniquely to a morphism $\Phi: \Sigma^* \to M$.

One branch of algebraic graph theory studies the connection between groups and corresponding Cayley graph representations of the groups. Similarly, monoids also have graph representations. Given a monoid $M$, we can construct a labeled multidigraph $G=(V,A)$ as follows. Let $V$ be the underlying set of the monoid and let $(m_1,m_2) \in A$ with label $m_3$ if $m_1m_3 = m_2$. See Figure \ref{cayley} for an example.

\begin{figure}
\begin{center}
\includegraphics[scale=0.8]{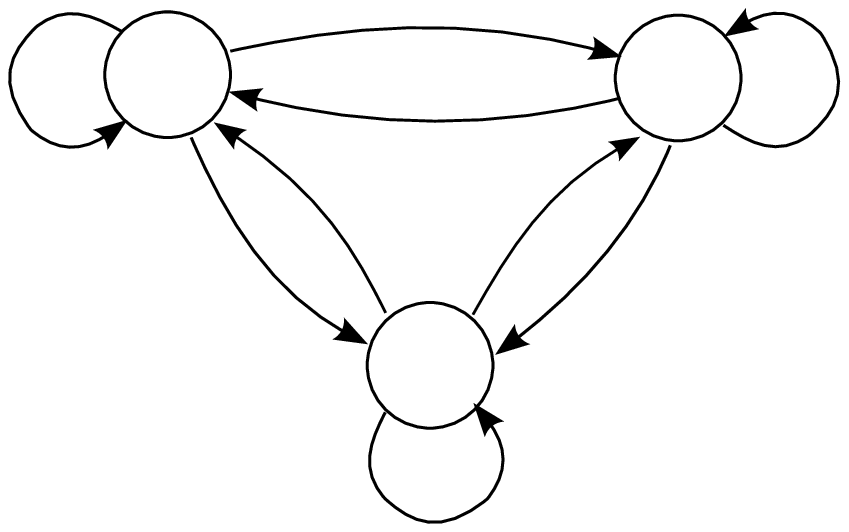}
\rput{0}(-4.1,4.8){1}
\rput{0}(-4.1,3.4){2}
\rput{0}(-6.25,4.1){$\mathbf{0}$}
\rput{0}(-7.75,4.1){0}
\rput{0}(-2.1,4.1){$\mathbf{1}$}
\rput{0}(-0.6,4.1){0}
\rput{0}(-4.1,1.75){$\mathbf{2}$}
\rput{0}(-4.1,0.2){0}
\rput{0}(-5.7,2.4){2}
\rput{0}(-2.5,2.4){1}
\rput{0}(-4.7,3){1}
\rput{0}(-3.5,3){2}
\end{center}
\caption{Graph of $(\mathbb{Z}_3, +)$.}
\label{cayley}
\end{figure}

Now the correspondence between monoids and automata should be clear since we can easily view the graph of $M$ as an automaton which recognizes a language over the alphabet $M$. All we need to do is declare the vertex $1_M$ as the initial state and agree upon a set of accepting vertices $F \subseteq M$. Observe that the graph of a monoid accepts a word $m_1m_2...m_n$ iff $m_1\cdot m_2 \cdot ... \cdot m_n \in F$. In fact, once we fix a function $\varphi : \Sigma \to M$, the graph of $M$ recognizes a language over the alphabet $\Sigma$: replace each arc's label by its preimage under $\varphi$ (now an arc can have more than one label). A word $s_1s_2...s_n \in \Sigma^*$ is accepted iff $\varphi(s_1) \cdot \varphi(s_2) \cdot ... \cdot \varphi(s_n) \in F$. 

If we allow the set of accepting states to vary and the function $\varphi : \Sigma \to M$ to vary (for fixed $\Sigma$ and $M$) then by viewing the monoid's graph as an automaton, we see that a single monoid can be used to recognize a family of languages over $\Sigma$. Each language in the family corresponds to a fixed set of accepting states and a fixed function $\varphi$. This leads to the more formal definition of recognition by a monoid. We say that a language $L \subseteq \Sigma^*$ is \emph{recognized} by a finite monoid $M$ if there exists a morphism $\Phi: \Sigma^* \to M$ and an accepting set $F \subseteq M$ such that $L = \Phi^{-1} (F)$. Similarly, we say that $\Phi: \Sigma^* \to M$ \emph{recognizes} $L$ if there exists $F \subseteq M$ such that $L = \Phi^{-1} (F)$. See Figure \ref{machine} for an alternative way of viewing $M$ as a language recognizer.

\begin{figure}
\begin{center}
\includegraphics[scale=1]{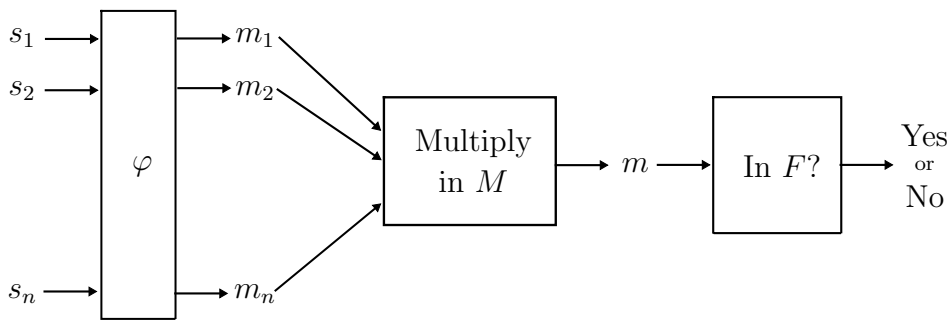}
\rput{0}(-11.4,2.5){$\mathbf{\varphi}$}
\rput{0}(-13,4.2){$s_1$}
\rput{0}(-13,3.5){$s_2$}
\rput{0}(-13,0.8){$s_n$}
\rput{0}(-9.9,4.2){$m_1$}
\rput{0}(-9.9,3.5){$m_2$}
\rput{0}(-9.9,0.8){$m_n$}
\rput{0}(-7,2.8){Multiply}
\rput{0}(-7,2.3){in $M$}
\rput{0}(-4.85,2.5){$m$}
\rput{0}(-2.9,2.5){In $F$?}
\rput{0}(-1,2.9){Yes}
\rput{0}(-1,2.5){\scriptsize{or}}
\rput{0}(-1,2.1){No}
\end{center}
\caption{Another way of viewing a monoid as a machine.}
\label{machine}
\end{figure}

Given any monoid morphism $\Psi: M \to N$, the \index{nuclear congruence} \emph{nuclear congruence} with respect to $\Psi$ is denoted by $\equiv_{\Psi}$ and is defined by $m \equiv_{\Psi} m'$ if $\Psi(m) = \Psi(m')$. We say that a set of words is \index{homogeneous set} \emph{homogeneous} with respect to $L$ if either every word in the set is in $L$ or none of the words is in $L$. Now observe that $\Phi: \Sigma^* \to M$ recognizes $L$ if and only if the nuclear congruence classes of $\Phi$ are homogeneous. This fact is used in the upcoming proofs.

From the earlier discussion, we can conclude that if $L$ is recognized by a finite monoid, then it is recognized by a finite automaton (the graph of the monoid) and therefore it is regular. In fact, the converse is also true.

\begin{theorem}\label{regular}
$L$ is regular if and only if a finite monoid recognizes $L$.
\end{theorem}
\begin{proof}
If $L$ is regular then it is recognized by a finite automaton. The definition of an automaton includes the transition function $\delta: Q \times \Sigma \to Q$ where $Q$ is the set of states. This function can be naturally extended to $\delta: Q \times \Sigma^* \to Q$. In other words, every word in $\Sigma^*$ defines a function from $Q$ to $Q$. Let $\delta_w : Q \to Q$, $q \mapsto \delta(q, w)$, be the function corresponding to the word $w$. Then it is easy to see that the set $T := \{\delta_w : w \in \Sigma^*\}$ is a monoid with the operation being composition of functions. Furthermore $T$ is finite since $Q$ is finite. We call $T$ the \index{transformation monoid} \emph{transformation monoid} of the automaton.

We claim that the transformation monoid $T$ recognizes $L$. To see this let $\Phi: \Sigma^* \to T$ be the canonical mapping: $w \mapsto \delta_w$. $\Phi$ is a morphism since
\[ \Phi(uv) = \delta_{uv} = \delta_u \circ \delta_v = \Phi(u) \circ \Phi(v).
\]
If $\Phi(u) = \Phi(v)$ then $u \in L$ iff $v \in L$ so the nuclear congruence classes are homogeneous and thus $\Phi$ recognizes $L$, which means $T$ recognizes $L$. 
\end{proof}

Theorem \ref{regular} constitutes the foundation of algebraic automata theory. It shows that finite monoids and finite automata have the same computational power with respect to the class of languages recognized. The proof reveals the strong link between monoids and automata. In fact, this link can be seen to be much stronger via the relation between the combinatorial properties\footnote{Regular languages are definable by regular expressions which are combinatorial descriptions of the language.} of $L$ and the algebraic properties of a monoid recognizing $L$. With the purpose of exploring this relation, we define the syntactic monoid of a regular language.

\subsection{The Syntactic Monoid}

For every regular language there is a minimal automaton that recognizes it. Similarly, every regular language has a ``minimal" monoid that recognizes it. We call this monoid the syntactic monoid and it is unique.

The \index{syntactic congruence} \emph{syntactic congruence} associated with a language $L \subseteq \Sigma^*$ is denoted by $\equiv_L$ and $x \equiv_L y$ if for all $u,v \in \Sigma^*$ we have $uxv \in L$ iff $uyv \in L$. It is straightforward to check that this relation is indeed a congruence. The \index{syntactic monoid} \emph{syntactic monoid} of $L$ is the quotient monoid $\Sigma^* / \equiv_L$ and is denoted by $M(L)$. Let $[w]$ represent the congruence class of $w$ with respect to the syntactic congruence. There is a well-defined operation: $[w][u] = [wu]$, so the canonical surjective mapping $\Phi: \Sigma^* \to M(L)$,  $w \mapsto [w]$, is a morphism. Observe that any congruence class of $\equiv_L$ is homogeneous (i.e. the nuclear congruence classes are homogeneous) so $\Phi$ recognizes $L$. We call $\Phi$ the \index{syntactic morphism}\emph{syntactic morphism}. 

It is quite easy to verify the following fact.

\begin{proposition}
$L$ is regular if and only if its syntactic monoid is finite.
\end{proposition}

We say that a monoid $N$ \index{division ! monoids}\emph{divides} a monoid $M$ (denoted by $N \prec M$) if there exists a surjective morphism from a submonoid\footnote{A submonoid is a subset that contains the identity and is closed under the operation.} of $M$ onto $N$. Intuitively, this means that the multiplicative structure of $N$ is embedded in $M$. The syntactic monoid of $L$ recognizes $L$ and is the minimal monoid with this property with respect to division.

\begin{proposition}\label{syntactic}
$M(L)$ recognizes $L$ and divides any other monoid that also recognizes $L$.
\end{proposition}
\begin{proof}
We have already proved that $M(L)$ recognizes $L$ so we prove the second statement.
	
Let $M$ be any monoid that recognizes $L$. So there exists a morphism $\Psi: \Sigma^* \to M$ recognizing $L$. Let $\Phi$ be the syntactic morphism. To show $M(L)$ divides $M$, we find a surjective morphism $\Upsilon$ from a submonoid $N$ of $M$ onto $M(L)$.

Before defining $\Upsilon$ we first prove the following claim: if $\Psi(a) = \Psi(b)$ then $\Phi(a) = \Phi(b)$. Suppose not, so there exists $a,b$ such that $\Psi(a) = \Psi(b)$ but $\Phi(a) \neq \Phi(b)$. By the definition of $\Phi$ this means that without loss of generality, there exists $u,v$ such that $uav \in L$ but $ubv \notin L$. We have
\[ \Psi(uav) = \Psi(u)\Psi(a)\Psi(v) = \Psi(u)\Psi(b)\Psi(v) = \Psi(ubv)
\]
so $\Psi$ maps $uav$ and $ubv$ to the same element. Since nuclear congruence classes (with respect to $\Psi$) must be homogeneous and $uav \in L$ but $ubv \notin L$, we get a contradiction.

Now let $N := \Psi(\Sigma^*)$. So $N$ is a submonoid of $M$. Define $\Upsilon: N \rightarrow M(L)$, $\Psi(w) \mapsto \Phi(w)$, i.e. $\Phi = \Upsilon \circ \Psi$.

\[
\begin{psmatrix}[mnode=R,colsep=2cm,rowsep=2cm]
\Sigma^* & M(L)\\
N\\
\end{psmatrix}
\psset{nodesep=0.3cm}
\everypsbox{\scriptstyle}
\ncLine{->}{1,1}{1,2}\Aput{\Phi}
\ncLine{->}{1,1}{2,1}\Bput{\Psi}
\ncLine{->}{2,1}{1,2}\Bput{\Upsilon}
\]

By claim $\Upsilon$ is well-defined. Since $\Phi$ is surjective, $\Upsilon$ is surjective. Furthermore,
\[ \Upsilon(\Psi(u)\Psi(v)) = \Upsilon(\Psi(uv)) = \Phi(uv) = \Phi(u)\Phi(v) = \Upsilon(\Psi(u))\Upsilon(\Psi(v))
\]
and so $\Upsilon$ is a morphism.
\end{proof}

If $M \prec N$ and $N \prec M$ then $M$ is isomorphic to $N$. So as claimed before, for every regular language there is a \emph{unique} (up to isomorphism) canonical monoid, the syntactic monoid, attached to it. An interesting property of $M(L)$ is that it is isomorphic to the transformation monoid of the minimal automaton recognizing $L$.


In the next subsection, we introduce the notion of language and monoid varieties. The combinatorial properties of a language are reflected on the algebraic properties of $M(L)$ and varieties are the proper framework to formalize this.

\subsection{Varieties}

We first give a brief overview of varieties and how monoid varieties and language varieties relate to each other.

A variety of languages is a family of languages that satisfy certain conditions. Similarly a variety of monoids is a family of monoids satisfying certain conditions. The variety theorem states that there is a one to one correspondence between varieties of regular languages and varieties of finite monoids: a variety of monoids $\mathbf{V}$ corresponds to the variety of regular languages $\mathcal{V}$ consisting of all the languages whose syntactic monoid is in $\mathbf{V}$. Consequently, we are able to state results of the form:

``A regular language belongs to the language variety $\mathcal{V}$ if and only if its syntactic monoid belongs to the monoid variety $\mathbf{V}$."

Many classes of languages that are defined combinatorially form language varieties and many classes of monoids that are defined algebraically form monoid varieties. So from above we can hope to reach results of the form:

``A regular language has the combinatorial property $P$ if and only if its syntactic monoid has the algebraic property $Q$."

Sch\"{u}tzenberger was the first to establish such a result: A regular language is star-free\footnote{A language is star-free if it can be defined by a extended regular expression without the Kleene star operation.} if and only if its syntactic monoid is finite and aperiodic\footnote{A monoid is aperiodic if no subset of it forms a non-trivial group.} (\cite{Sch65}). Several important classes of regular languages admit a similar algebraic characterization. This often yields decidability results which are not known to be obtainable by other means. For instance, by a result of McNaughton and Papert (\cite{MP71}), we know that regular languages definable by a first-order formula are exactly the star-free languages. This implies that we can decide if a regular language is first-order definable by checking if its syntactic monoid is aperiodic and this is the only known way of doing this. These types of algebraic characterizations of regular languages also provide one with powerful algebraic tools when analyzing and proving results about regular languages.

\subsubsection{Varieties of Finite Monoids}

A \index{variety ! of monoids} \emph{variety of finite monoids} is a family of finite monoids $\mathbf{V}$ that satisfies the following two conditions:
\begin{itemize}
\item[(i)] $\mathbf{V}$ is closed under division: if $M \in \mathbf{V}$ and $N \prec M$ then $N \in \mathbf{V}$,
\item[(ii)] $\mathbf{V}$ is closed under direct product: if $M_1, M_2 \in \mathbf{V}$ then $M_1 \times M_2 \in \mathbf{V}$.
\end{itemize}


\begin{example}
The following are some examples of varieties of monoids:
\begin{itemize}
\item $\mathbf{I}$ is the trivial variety consisting of only the trivial monoid $I = \{1\}$.
\item $\mathbf{M}$ is the variety containing all finite monoids.
\item $\mathbf{Com}$ is the variety of all commutative monoids.
\item $\mathbf{G}$ is the variety of groups.
\item $\mathbf{A}$ is the variety of aperiodic monoids.
\item $\mathbf{J}$ is the variety of monoids $M$ that satisfy $Mm_1M = Mm_2M \implies m_1= m_2$. We call these monoids $\mathcal{J}$-trivial.
\end{itemize}
\end{example}

There is a convenient way of defining varieties of monoids through \index{identities} \emph{identities}. The notion of identities can be presented in two ways. One involves topological semigroups (see \cite{Pin97}), which we wish to avoid. Therefore we use the presentation which we think is more intuitive.

Let $\Sigma$ be a countable alphabet and $u,v$ two words in $\Sigma^*$. We say that a monoid $M$ \emph{satisfies} the identity $u = v$ if for all morphisms $\varphi : \Sigma^* \to M$ we have $\varphi(u) = \varphi(v)$. This means that if we replace the letters of $u$ and $v$ with arbitrary (but consistent) elements of $M$ then we will arrive at an equality in $M$. For example a monoid is commutative if and only if it satisfies the identity $ab = ba$.

It can be shown that the family of finite monoids consisting of the monoids that satisfy the identity $u = v$ forms a variety. This variety is denoted by $\mathbf{V}(u,v)$.

Now let $(u_n,v_n)_{n>0}$ be a sequence of pair of words in $\Sigma^*$. Define
\[
\mathbf{W} := \lim_{n \to \infty} \mathbf{V}(u_n, v_n) = \bigcup_{m>0} \bigcap_{n \geq m} \mathbf{V}(u_n, v_n).
\]
Observe that $M \in \mathbf{W}$ if and only if there exists an $n_0 > 0$ such that for every $n > n_0$, $M$ satisfies $u_n = v_n$. Here we say that $\mathbf{W}$ is \index{ultimately defined ! unordered monoids} \emph{ultimately defined} by the sequence of identities $(u_n = v_n)_{n>0}$.

\begin{theorem}[see \cite{Pin86}]
Every variety of monoids is ultimately defined by some sequence of equations.
\end{theorem}

For example the variety of commutative monoids is ultimately defined by the constant sequence $(ab = ba)_{n>0}$. A less trivial example is the variety of aperiodic monoids. It can be shown that a finite monoid is aperiodic if and only if for each $m \in M$ there exists $n \geq 0$ such that $m^n = m^{n+1}$. Consequently, the variety of aperiodic monoids is ultimately defined by the sequence $(a^n = a^{n+1})_{n>0}$. The variety of commutative aperiodic monoids is ultimately defined by the sequence
\[ a^1 = a^2, ab = ba, a^2 = a^3, ab = ba, a^3 = a^4, ab = ba, ... \]
In such a case, for clarity, we say that the variety is ultimately defined by two sequences.

An element $m$ of a monoid is called \index{idempotent} \emph{idempotent} if $m \cdot m = m$. In finite monoids, idempotents play a key role. For instance, every non-empty monoid contains an idempotent. Indeed, if we take any element $m$ of the monoid, then there exists a number $n > 0$ such that $m^n$ is an idempotent (in fact this is the unique idempotent generated by $m$). This implies that for any finite monoid, there is a number $k >0$ such that for every element $m$ in the monoid, we have that $m^k$ is an idempotent. We call $k$ an \index{exponent} \emph{exponent} of $M$. Observe that if $k$ is an exponent of $M$ then for any $n \geq k$, $n!$ is also an exponent of $M$.

We use $n!$ in many sequences of identities that ultimately define varieties of monoids. For example, the sequence $(x^{n!} y x^{n!} = x^{n!})_{n>0}$ ultimately defines the variety of \emph{locally trivial monoids}. From this, it should be clear that a monoid $M$ is locally trivial if and only if for every idempotent $e \in M$ and every element $m \in M$ we have $eme = e$. As a convention, a sequence of equations involving $n!$ is written by replacing $n!$ with $\omega$. So for example we use $x^{\omega} y x^{\omega} = x^{\omega}$ as an abbreviation for $(x^{n!} y x^{n!} = x^{n!})_{n>0}$. It is easy to see that the variety of groups $\mathbf{G}$ is ultimately defined by $x^{\omega} = \epsilon$.

Given a sequence of identities $E$, we denote by $[[E]]$ the variety that is ultimately defined by $E$. So for example we have $\mathbf{G} = [[x^{\omega} = \epsilon]]$ and the variety of locally trivial monoids is $[[ x^{\omega} y x^{\omega} = x^{\omega} ]]$.

\subsubsection{Varieties of Regular Languages}

Before we can define a variety of regular languages we need some preliminary definitions.

A \emph{class of regular languages} is a function $\mathcal{C}$ that maps every alphabet $\Sigma$ to a set of regular languages in $\Sigma^*$.

A set of languages in $\Sigma^*$ that is closed under finite intersection, finite union and complementation is called a \index{boolean algebra} \emph{boolean algebra}.


Now a \index{variety ! of regular languages} \emph{variety of regular languages} is a class of regular languages $\mathcal{V}$ that satisfies the following conditions:
\begin{itemize}
\item[(i)] For any alphabet $\Sigma$, $\mathcal{V}(\Sigma)$ is a boolean algebra.
\item[(ii)] $\mathcal{V}$ is closed under inverse morphisms: given any alphabets $\Sigma$ and $\Gamma$, for any morphism $\Phi : \Sigma^* \to \Gamma^*$, if $L \in \mathcal{V}(\Gamma)$ then $\Phi^{-1}(L) \in \mathcal{V}(\Sigma)$.
\item[(iii)] $\mathcal{V}$ is closed under left and right quotients: for $L \in \mathcal{V}(\Sigma)$ and $s \in \Sigma$, we have $s^{-1}L := \{w \in L | sw \in L \}$ and $Ls^{-1} := \{ w \in L | ws \in L\}$ are in $\mathcal{V}(\Sigma)$.
\end{itemize}

\begin{example}
The following are some examples of varieties of regular languages:
\begin{itemize}
\item The trivial variety: $\mathcal{V}(\Sigma) = \{\emptyset, \Sigma^*\}$.
\item The variety of all regular languages (each alphabet is mapped to all the regular languages over this alphabet).
\item The variety of star-free languages.
\item The variety of piecewise testable languages: A language is called \emph{piecewise testable} if there exists a $k \in \mathbb{N}$ such that membership of any word in the language depends on the set of subwords\footnote{A word $u = a_1...a_n$ is a subword of a word $x$ if $x = x_0a_1x_1a_2...a_nx_n$ for some words $x_1,...,x_n$.} of length at most $k$ occurring in that word. 
\end{itemize}
\end{example}

\subsubsection{The Variety Theorem}

For a given finite monoid variety $\mathbf{V}$, let $\mathcal{V}(\Sigma)$ be the set of languages in $\Sigma^*$ whose syntactic monoid belongs to $\mathbf{V}$. Alternatively, we can define $\mathcal{V}$ as follows.

\begin{proposition}
Let $\mathcal{C}(\Sigma)$ be the set of languages over $\Sigma$ that is recognized by a monoid in $\mathbf{V}$. Then $\mathcal{V} = \mathcal{C}$. 
\end{proposition}
\begin{proof}
$\mathcal{V} \subseteq \mathcal{C}:$ If $L \in \mathcal{V}$ then $M(L) \in \mathbf{V}$. $M(L)$ recognizes $L$ so $L \in \mathcal{C}$.\\
$\mathcal{C} \subseteq \mathcal{V}:$ If $L \in \mathcal{C}$ then there exists $M \in \mathbf{V}$ recognizing $L$. $M(L) \prec M$ and $\mathbf{V}$ is closed under division so $M(L) \in \mathbf{V}$. Therefore $L \in \mathcal{V}$.
\end{proof}

Now we can state the variety theorem due to Eilenberg (\cite{Eil74}).

\begin{theorem}[The Variety Theorem]
$\mathcal{V}$ is a variety of languages and the mapping $\mathbf{V} \mapsto \mathcal{V}$ is one to one.
\end{theorem}

In light of this theorem, one can hope to explicitly make such correspondences. Two important correspondence results are the following.

\begin{theorem}[\cite{Sch65}]
The monoid variety $\mathbf{A}$ corresponds to the variety of star-free languages. Equivalently, a regular language is star-free if and only if its syntactic monoid is aperiodic.
\end{theorem}

\begin{theorem}[\cite{Sim75}]
The monoid variety $\mathbf{J}$ corresponds to the variety of piecewise testable languages. Equivalently, a regular language is piecewise testable if and only if its syntactic monoid is $\mathcal{J}$-trivial.
\end{theorem}

Furthermore we can restate Theorem \ref{regular} as follows.

\begin{theorem}
The monoid variety $\mathbf{M}$ corresponds to the variety of all regular languages. Equivalently, $L$ is regular if and only if its syntactic monoid is finite.
\end{theorem}


\section{Ordered Monoids}

In the previous section, we have seen that we can classify regular languages in terms of the monoids that recognize them. We were able to obtain algebraic characterizations for certain classes of languages: varieties of languages. Many interesting combinatorially defined classes of languages form varieties. But there are other combinatorially defined classes of languages that do not form a variety. Of particular interest are families of languages that are not necessarily closed under complementation but satisfy the other properties of a variety. We call such families ``positive varieties of languages". Is it possible to get a similar algebraic characterization for these languages as well? In particular, is there a result similar to Eilenberg's variety theorem that permits us to treat positive varieties?

Fortunately the answers to the above questions are ``yes". The idea is to attach an order on the monoids and adapt the definition of recognition by monoids to ordered monoids. This point of view is a generalization of the unordered case and allows us to make a one to one correspondence between varieties of ordered monoids and positive varieties of languages. This extension was introduced in \cite{Pin95}.

Intuitively speaking, the syntactic monoid has less information than the minimal automaton. One reason for this is that in the minimal automaton the accepting states are predetermined, but in the syntactic monoid the accepting set is not. As we see in the next subsection, the order on the monoid restricts the way we can choose the accepting set and consequently the ordered syntactic monoid recovers some of the missing information. This restriction lets us analyze classes of languages that are not closed under complementation.

In this section, we go over the definitions and the results seen thus far, and present the analogous ordered counterparts. We start with the notion of recognition by ordered monoids. Then we define the syntactic ordered monoid. Later we look at varieties of ordered monoids, positive varieties of languages and the variety theorem that establishes a one to one correspondence between these ordered monoid varieties and positive language varieties.

\subsection{Recognition by Ordered Monoids}

An \emph{order} relation on a set $S$ is a relation that is reflexive, anti-symmetric and transitive and it is denoted by $\leq$. We say that $\leq$ is a \emph{stable order relation} on a monoid $M$ if for all $x,y,z \in M$, $x \leq y$ implies $zx \leq zy$ and $xz \leq yz$.

An \index{ordered monoid} \emph{ordered monoid} $(M, \leq_M)$ is a monoid $M$ together with a stable order relation $\leq_M$ that is defined on $M$. A \index{morphism ! between ordered monoids} \emph{morphism of ordered monoids} $\Phi: (M, \leq_M) \to (N, \leq_N)$ is a morphism between $M$ and $N$ that also preserves the order relation, i.e. for all $m,m' \in M$, $m \leq_M m'$ implies $\Phi(m) \leq_N \Phi(m')$.

The free monoid $\Sigma^*$ will always be equipped with the equality relation. Observe that any morphism $\Phi: \Sigma^* \to M$ is also a morphism of ordered monoids $\Phi: (\Sigma^*,=) \to (M, \leq_M)$ for any stable order $\leq_M$ and vice versa.

A subset $I \subseteq M$ is called an \index{order ideal}\emph{order ideal} if for any $y \in I$, $x \leq_M y$ implies $x \in I$. Observe that every order ideal $I$ in a finite monoid $M$ has a generating set $x_1,...,x_k$ such that $I = \langle x_1,...,x_k \rangle := \{y \in M : \exists x_i \textrm{ with } y \leq_M x_i \}$. 

Now the concept of recognizability is very similar to the unordered case. We say that a language $L \subseteq \Sigma^*$ is recognized by an ordered monoid $(M, \leq_M)$ if there exists a morphism of ordered monoids $\Phi : (\Sigma^*, =) \to (M, \leq_M)$ and an order ideal $I \subseteq M$ such that $L = \Phi^{-1}(I)$. Equivalently, $L$ is recognized by $(M, \leq_M)$ if there exists a morphism $\Phi : \Sigma^* \to M$ and an order ideal $I \subseteq M$ such that $L = \Phi^{-1}(I)$. Observe that this is a generalization of the unordered case in the sense that any monoid is an ordered monoid with the equality order (the trivial order) and any subset of the monoid is an order ideal with respect to equality. Also note that in the unordered case, if $L$ is recognized by $M$, then so is the complement of $L$. In the ordered case, since we require the accepting set to be an order ideal, this statement is no longer true. This restriction on the accepting set allows the ordered monoid to keep more information about the automaton recognizing $L$. In this sense, one can think of the ordered case as a refinement of the unordered case.




\subsection{The Syntactic Ordered Monoid}

The definition of the syntactic congruence with respect to $L$ is as exactly as before: $x \equiv_L y$ if for all $u,v \in \Sigma^*$ we have $uxv \in L$ iff $uyv \in L$. Also the syntactic monoid is the quotient monoid $M(L) = \Sigma^* / \equiv_L$. To be able to get a similar variety theorem for classes of languages not closed under complementation, we need to define a stable order on $M(L)$ that allows us to obtain an ordered counterpart of the variety theorem.

First, break up $\equiv_L$: Let $x \preceq_L y$ if for all $u,v \in \Sigma^*$, $uyv \in L \implies uxv \in L$. So $x \equiv_L y$ if and only if $x \preceq_L y$ and $y \preceq_L x$. Now $\preceq_L$ induces a well-defined stable order $\leq_L$ on $M(L)$ given by
\[ [x] \leq_L [y] \textrm{ if and only if } x \preceq_L y. \]
It is straightforward to check that this is indeed a well-defined stable order. The ordered monoid $(M(L), \leq_L)$ is the \index{syntactic ordered monoid} \emph{syntactic ordered monoid} of $L$.

We say that an ordered monoid $(N, \leq_N)$ \index{division ! ordered monoids} \emph{divides} an ordered monoid $(M, \leq_M)$ if there exists a surjective morphism of ordered monoids from a submonoid\footnote{A submonoid of $(M, \leq_M)$ is a submonoid of $M$ with the order being the restriction of $\leq_M$ to the submonoid.} of $(M, \leq_M)$ onto $(N, \leq_N)$.

Now we state and prove the analog of Proposition \ref{syntactic}. We give the proof to demonstrate that slight modifications to the original proof suffices to obtain the proof for the ordered counterpart.

\begin{proposition}
$(M(L), \leq_L)$ recognizes $L$ and is the minimal ordered monoid with this property.
\end{proposition}
\begin{proof}
Let $\Phi: \Sigma^* \to M(L)$ be the surjective canonical mapping: $w \mapsto [w]$, i.e. $\Phi$ is the syntactic morphism. We already know that the congruence classes are homogeneous so all we need to show is that the accepting set $I$ is an order ideal, i.e. we need to show that if $[y] \in I$ and $[x] \leq_L [y]$, then $[x] \in I$. Since $[x] \leq_L [y]$, we have $x \preceq_L y$ and so for all $u,v \in \Sigma^*$, $uyv \in L \implies uxv \in L$. In particular $y \in L \implies x \in L$. Since $[y] \in I$, $y \in L$ and therefore $x \in L$ and so $[x] \in I$ as required.

Let $(M, \leq_M)$ be any monoid that recognizes $L$. So there exists a morphism $\Psi: \Sigma^* \to M$ and an order ideal $I \subseteq M$ such that $L = \Psi^{-1}(I)$. Let $\Phi$ be defined as above. To show $(M(L), \leq_L)$ divides $(M, \leq_M)$ we find a surjective morphism of ordered monoids $\Upsilon$ from a submonoid of $(M, \leq_M)$ onto $(M(L), \leq_L)$.

We let $N := \Psi(M)$ so $(N, \leq_M)$ is a submonoid of $(M,\leq_M)$. Define $\Upsilon$ to be the same function as the one we defined in the proof of Proposition \ref{syntactic}, so $\Upsilon$ is such that $\Upsilon: (N, \leq_M) \rightarrow (M(L), \leq_L)$, $\Psi(w) \mapsto \Phi(w)$, i.e. $\Phi = \Upsilon \circ \Psi$. As shown before, $\Upsilon$ is a well-defined surjective morphism. What remains to be shown is that $\Upsilon$ is a morphism of ordered monoids. For this, we need to show $\Psi(a) \leq_M \Psi(b) \implies \Upsilon(\Psi(a)) \leq_L \Upsilon(\Psi(b))$, i.e. $\Psi(a) \leq_M \Psi(b) \implies \Phi(a) \leq_L \Phi(b)$.

Suppose the above is not true. So there exists $a$ and $b$ such that  $\Psi(a) \leq_M \Psi(b)$ but $\Phi(a) \nleq_L \Phi(b)$. This means that $a \npreceq_L b$ and therefore there exists $u,v \in \Sigma^*$ such that $ubv \in L$ but $uav \notin L$. On the other hand, since $\leq_M$ is a stable order we have
\[ \Psi(uav) = \Psi(u)\Psi(a)\Psi(v) \leq_M \Psi(u)\Psi(b)\Psi(v) = \Psi(ubv). \]
$ubv \in L$ implies that $\Psi(ubv) \in I$ and by above and the fact that $I$ is an order ideal we must have that $\Psi(uav) \in I$. This is a contradiction since $uav \notin L$.
\end{proof}

\subsection{Varieties}

The definition of an ordered monoid variety is identical to the unordered case. We say that a family of ordered monoids $\mathbf{V}$ is a \index{variety ! of ordered monoids} \emph{variety of ordered monoids} if it is closed under division of ordered monoids and finite direct product\footnote{The order in a finite direct product $M_1 \times ... \times M_n$ is given by $(m_1,...,m_n) \leq (m_1',...,m_n')$ iff $m_i \leq m_i' \quad \forall i \in [n]$.}.

Similar to unordered monoid varieties, varieties of ordered monoids can be defined using identities. We say that $(M, \leq_M)$ \emph{satisfies} the identity $u \leq v$ if and only if for every morphism $\varphi : \Sigma^* \to M$ we have $\varphi(u) \leq_M \varphi(v)$. Let $\mathbf{V}(u,v)$ be the variety of ordered monoids that satisfy the identity $u \leq v$. Then given a sequence of pair of words $(u_n,v_n)_{n>0}$, $\mathbf{W} := \lim \mathbf{V}(u_n,v_n)$ is said to be \index{ultimately defined ! ordered monoids} \emph{ultimately defined} by this sequence.

\begin{theorem}[\cite{PW96}]
Every variety of ordered monoids is ultimately defined by some sequence of identities.
\end{theorem}

Now we define positive variety of languages. A set of languages in $\Sigma^*$ that is closed under finite intersection and finite union is called a \index{positive boolean algebra} \emph{positive boolean algebra}. So it differs from a boolean algebra because we do not require the set to be closed under complementation. A class of languages $\mathcal{V}$ is called a \index{variety ! positive} \emph{positive variety of languages} if it is a positive boolean algebra, is closed under inverse morphisms and is closed under left and right quotients. 

For a given variety of finite ordered monoids $\mathbf{V}$, let $\mathcal{V}(\Sigma)$ be the set of languages over $\Sigma$ whose syntactic ordered monoid belongs to $\mathbf{V}$. As before, this is equivalent to saying that $\mathcal{V}(\Sigma)$ is the set of languages over $\Sigma$ that are recognized by an ordered monoid in $\mathbf{V}$.

\begin{theorem}[The Variety Theorem \cite{Pin95}]
$\mathcal{V}$ is a positive variety of languages and the mapping $\mathbf{V} \mapsto \mathcal{V}$ is one to one.
\end{theorem}

Now we give two explicit correspondences. The interested reader can find the proofs in \cite{Pin95}.

Let $\Gamma$ be a subset of the alphabet $\Sigma$. Define $L(\Gamma)$ as
\[ L(\Gamma) := \bigcap_{a \in \Gamma} \Sigma^* a \Sigma^*. \]
This is equivalent to saying that $L(\Gamma)$ is the set of words that contain at least one occurrence of each letter in $\Gamma$.

A monoid is idempotent if every element in the monoid is idempotent.

\begin{theorem}
A language in $\Sigma^*$ is a finite union of languages of the form $L(\Gamma)$ for $\Gamma \subseteq \Sigma$ if and only if it is recognized by a finite commutative idempotent ordered monoid $(M, \leq_M)$ in which the identity is the greatest element with respect to the order.
\end{theorem}

A language $L$ is a \emph{shuffle ideal} if it satisfies the following property: if a word $w$ has a subword in $L$, then $w$ is in $L$.

\begin{theorem}
A language is a shuffle ideal if and only if it is recognized by a finite ordered monoid in which the identity is the greatest element.
\end{theorem}

We conclude this chapter by pointing out that our main interest is in positive varieties of languages (and consequently in ordered monoids) because regular languages having $O(f)$ non-deterministic communication complexity form a positive variety of languages (see next chapter). For the communication models studied in \cite{TT03}, regular languages having $O(f)$ communication complexity form a variety of languages and so the theory of ordered monoids is not necessary. From now on, we abandon unordered monoids and work with the more general theory of ordered monoids.


%% file: chap4.tex
\chapter{Communication Complexity of Regular Languages}\label{complexityofregularlanguages}

The main goal of this chapter is to prove upper and lower bounds on the non-deterministic communication complexity of regular languages. In Section 4.1, we formally define the communication complexity of finite ordered monoids and regular languages. We prove two theorems that establish the soundness of an algebraic approach to the communication complexity of regular languages. In Section 4.2, we present a form of the definition of rectangular reductions and introduce local rectangular reductions. Then we present upper and lower bound results for regular languages in which the lower bounds are established using rectangular reductions from three functions we have seen in Chapter 2. We also state an intriguing conjecture that gives an exact characterization of the non-deterministic communication complexity of regular languages.

\section{Algebraic Approach to Communication Complexity}

In Chapter 2, we studied the communication complexity of functions that have 2 explicit inputs, each being an $n$-bit string. In order to define the communication complexity of a monoid and a regular language, we need to generalize the definition of communication complexity to include functions that have a single input string. Suppose a function $f$ has one $n$-bit string $x_1...x_n$ as input and let $A \cup B$ be a partition of $[n]$. Then the communication complexity of $f$ with respect to this partition is the communication complexity of $f$ when Alice receives the bits $x_i$ for all $i \in A$ and Bob receives the bits $x_j$ for all $j \in B$. For instance, in the non-deterministic model we denote this by $N^1_{AB} (f)$. In this case, the non-deterministic communication complexity of $f$ is defined as
\[ N^1(f) := \max_{A,B} N^1_{AB}(f) \]
where the maximum is taken over all possible partitions of $[n]$. The partition that achieves this maximum is called a \emph{worst case partition}.

Note that the communication complexity definitions and results seen thus far apply to functions that have inputs that are strings of length $n$ over any fixed alphabet. That is, the requirement of bit strings as inputs can be relaxed.

\index{communication complexity ! of a monoid}
We define the communication complexity of a finite ordered monoid using the worst-case partitioning notion. The communication complexity of a pair $(M,I)$ where $M$ is a finite ordered monoid and $I$ is an order ideal in $M$ is the communication complexity of the monoid evaluation problem corresponding to $M$ and $I$: Alice is given $m_1,m_3,...,m_{2n-1}$ and Bob is given $m_2,m_4,...,m_{2n}$ such that each $m_i \in M$. They want to decide if the product $m_1m_2...m_{2n}$ is in $I$. The communication complexity of $M$ is the maximum communication complexity of $(M,I)$ where $I$ ranges over all order ideals in $M$. Observe that if for example Alice were to receive $m_i$ and $m_{i+1}$, then she could multiply these monoid elements and treat them as one monoid element. This is why for a worst-case partition, Alice and Bob should not get consecutive monoid elements.


\index{communication complexity ! of a language}
Similarly, we define the communication complexity of a regular language $L \subseteq \Sigma^*$ as the communication complexity of the language problem corresponding to $L$: Alice is given $a_1, a_3, ..., a_{2n-1}$ and Bob is given $a_2, a_4, ... , a_{2n}$ such that each $a_i \in \Sigma \cup \{\epsilon\}$ where $\epsilon$ represents the empty word in $\Sigma^*$ (also referred to as the empty letter). They want to determine if $a_1a_2...a_{2n} \in L$. The way the input is distributed corresponds to the worst-case partition since we allow $a_i$ to be empty letters.

As mentioned in Chapter 1, our aim is to find functions $f_1(n), ... , f_k(n)$ such that each regular language has $\Theta (f_i(n))$ non-deterministic communication complexity for some $i \in \{1,2,...,k\}$. We would also like a characterization of the languages with $\Theta(f_i(n))$ complexity for all $i \in \{1,2,...,k\}$. The next two results show that such a characterization can be obtained by looking at the algebraic properties of regular languages.

\begin{theorem}\label{equalcomplexity}
Let $L \subseteq \Sigma^*$ be a regular language with $M(L) = M$. We have $N^1(M) = \Theta(N^1(L))$.
\end{theorem}

\begin{theorem}\label{formsvariety}
For any increasing function $f: \mathbb{N} \to \mathbb{N}$, the class of ordered monoids $\mathbf{V}$ such that each monoid $M \in \mathbf{V}$ satisfies $N^1(M)=O(f)$ forms a variety of ordered monoids.
\end{theorem}

These two theorems together with the variety theorem imply that the class of languages $\mathcal{V}$ such that for any $L \in \mathcal{V}$ we have $N^1(L) = O(f)$ forms a positive variety of languages. So a characterization in terms of positive varieties is possible. Furthermore, observe that communication complexity of monoids parametrize the communication complexity of regular languages. Bounds on monoids yield bounds on regular languages and vice versa. When proving such bounds, carefully choosing between the two directions can considerably simplify the analysis. Usually we find that upper bound arguments are easier to establish with the combinatorial descriptions of a language whereas lower bound arguments are easier to establish with the algebraic descriptions of the corresponding syntactic monoid.

\begin{proof}[Proof of Theorem \ref{equalcomplexity}]
First we show that $N^1(L) = O(N^1(M))$. For this, we present a non-deterministic protocol for $L$. Suppose Alice is given $a_1a_2...a_n$ and Bob is given $b_1b_2...b_n$. Let $\Phi$ be the syntactic morphism and let $I$ be the accepting order ideal. The protocol is as follows: Alice computes $\Phi(a_1),...,\Phi(a_n)$ and Bob computes $\Phi(b_1),...,\Phi(b_n)$. Using the protocol for the monoid evaluation problem of $(M,I)$, they can decide at $O(N^1(M))$ cost if 
\[\Phi(a_1)\Phi(b_1)...\Phi(a_n)\Phi(b_n) = \Phi(a_1b_1...a_nb_n) \in I. \] 
This determines if $a_1b_1...a_nb_n$ is in $L$ or not.

Now we show that $N^1(M) = O(N^1(L))$. We present a protocol for $(M,I)$ where $I = \langle i_1,...,i_k \rangle$ is some order ideal in $M$. Before presenting the protocol, we first fix some notation and definitions. Again let $\Phi$ be the syntactic morphism. For each monoid element $m$, fix a word that is in the preimage of $m$ under $\Phi$, and denote it by $w_m$. Let $Y_a := \{(u,v) : uav \in L\}$. Recall that $a \preceq_L b$ if for all $u,v \in \Sigma^*$, $ubv \in L \implies uav \in L$. So 
\[ \Phi(a) \leq_L \Phi(b) \textrm{ iff } a \preceq_L b \textrm{ iff } Y_b \subseteq Y_a.\]
For each $Y_a$ and $Y_b$ with $Y_b \nsubseteq Y_a$, pick $(u,v)$ such that $(u,v) \in Y_b$ but $(u,v) \notin Y_a$. Let $K$ be the set of all these $(u,v)$. One can think of $K$ as containing a witness for $Y_b \nsubseteq Y_a$ for each such pair. Note that $K$ is finite. Now pad each $w_m$ and each word appearing in a pair in $K$ with the empty letter $\epsilon$ so that each of these words have the same length. Observe that this length is a constant that does not depend on the length of the input that Alice and Bob will receive.

Now assuming that Alice and Bob have agreed upon the definitions made thus far, the protocol is as follows. Suppose Alice is given $m^a_1,m^a_2,...,m^a_n$ and Bob is given $m^b_1,m^b_2,...,m^b_n$. For each $i_j$ they want to determine if $m^a_1 m^b_1 ... m^a_n m^b_n \leq_L i_j$. This is equivalent to determining if $w_{m^a_1 m^b_1 ... m^a_n m^b_n} \preceq_L w_{i_j}$, and this is equivalent to $w_{m^a_1} w_{m^b_1} ... w_{m^a_n} w_{m^b_n} \preceq_L w_{i_j}$. If this is not the case, then $Y_{w_{i_j}} \nsubseteq Y_{w_{m^a_1} w_{m^b_1} ... w_{m^a_n} w_{m^b_n}}$ and so there will be a witness of this in $K$, i.e. there exists $(u,v)$ such that $u w_{i_j} v \in L$ but $ u w_{m^a_1} w_{m^b_1} ... w_{m^a_n} w_{m^b_n} v \notin L$. If indeed $w_{m^a_1} w_{m^b_1} ... w_{m^a_n} w_{m^b_n} \preceq_L w_{i_j}$ then for each $(u,v) \in K$ with $u w_{i_j} v \in L$, we will have $u w_{m^a_1} w_{m^b_1} ... w_{m^a_n} w_{m^b_n} v \in L$. Using the protocol for $L$, Alice and Bob can check which of the two cases is true. The following shows how Alice and Bob's inputs look like before running the protocol for $L$. Note that each block has the same constant length.
\begin{center}
\includegraphics[scale=1]{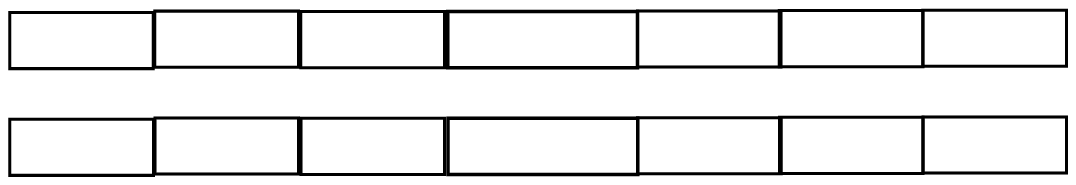}
\rput{0}(-12.5,1.75){Alice}
\rput{0}(-12.5,0.68){Bob}
\rput{0}(-10.7,1.75){$u$}
\rput{0}(-9.2,1.68){$w_{m_1^a}$}
\rput{0}(-7.75,1.75){$\epsilon \epsilon... \epsilon$}
\rput{0}(-6,1.75){...}
\rput{0}(-4.35,1.68){$w_{m_n^a}$}
\rput{0}(-2.9,1.75){$\epsilon \epsilon ... \epsilon$}
\rput{0}(-1.5,1.75){$v$}
\rput{0}(-10.7,0.68){$\epsilon \epsilon ... \epsilon$}
\rput{0}(-9.2,0.68){$\epsilon \epsilon ... \epsilon$}
\rput{0}(-7.75,0.62){$w_{m_1^b}$}
\rput{0}(-6,0.68){...}
\rput{0}(-4.35,0.68){$\epsilon \epsilon ... \epsilon$}
\rput{0}(-2.9,0.62){$w_{m_n^b}$}
\rput{0}(-1.5,0.68){$\epsilon \epsilon ... \epsilon$}
\end{center}
\end{proof}

The proof of Theorem \ref{formsvariety} follows from the following two lemmas. The first lemma shows that $\mathbf{V}$ is closed under finite direct product. The second lemma shows that $\mathbf{V}$ is closed under division of monoids.

\begin{lemma}
Let $(M,\leq_M)$ and $(N,\leq_N)$ be ordered monoids. Then $N^1(M \times N) \leq N^1(M) + N^1(N)$.
\end{lemma}
\begin{proof}
Any order ideal in $M \times N$ will be of the form $I \times J$ where $I$ is an order ideal in $M$ and $J$ is an order ideal in $N$. Therefore testing whether a product of elements in $M \times N$ is in an order ideal $I \times J$ or not can be done by testing if the product of the first coordinate elements is in $I$ and testing if the product of the second coordinate elements is in $J$.
\end{proof}

\begin{lemma}\label{division}
Let $(M,\leq_M)$ and $(N,\leq_N)$ be ordered monoids such that $N \prec M$. Then $N^1(N) \leq N^1(M)$.
\end{lemma}
\begin{proof}
Since $N \prec M$, there is a surjective morphism $\phi$ from a submonoid $M'$ of $M$ onto $N$. Denote by $\phi^{-1}(n)$ a fixed element from the preimage of $n$.

Let $I$ be an order ideal in $N$. A protocol for $(N,I)$ is as follows. Alice is given $n_1^a, n_2^a,..., n_t^a$ and Bob is given $n_1^b, n_2^b,..., n_t^b$. They want to decide if $n_1^a n_1^b ... n_t^a n_t^b \in I$. This is equivalent to deciding if 
\[\phi^{-1}(n_1^a) \phi^{-1}(n_1^b) ... \phi^{-1}(n_t^a)\phi^{-1}(n_t^b) \in \phi^{-1}(I). \] 
It can be easily seen that $\phi^{-1}(I)$ is an order ideal in $M'$ so Alice and Bob can use the protocol for $M'$ to decide if the above is true. Therefore we have $N^1(N) \leq N^1(M')$. It is straightforward to check that $N^1(M') \leq N^1(M)$ and so $N^1(N) \leq N^1(M)$ as required.
\end{proof}

\section{Complexity Bounds for Regular Languages and Monoids}

In this section, we present upper and lower bounds for the non-deterministic communication complexity of certain classes of languages. Upper bounds are established by presenting an appropriate protocol whereas lower bound arguments are based on rectangular reductions from the following functions: LESS-THAN, PROMISE-DISJOINTNESS, INNER-PRODUCT. In Chapter 2, we have seen that each of these functions require linear communication in the non-deterministic model. We have also seen the definition of a rectangular reduction. We give here a form of this definition which specifically suits our needs in this section.

\begin{definition}
Let $f: \{0,1\}^n \times \{0,1\}^n \to \{0,1\}$, $M$ a finite ordered monoid and $I$ an order ideal in $M$. A \index{rectangular reduction} \emph{rectangular reduction} of length $t$ from $f$ to $(M,I)$ is a sequence of $2t$ functions $a_1,b_2,a_3,...,a_{2t-1},b_{2t}$ with $a_i : \{0,1\}^n \to M$ and $b_i : \{0,1\}^n \to M$ and such that for every $x,y \in \{0,1\}^n$ we have $f(x,y) = 1$ if and only if the product $a_1(x)b_2(y)...b_{2t}(y)$ is in $I$.

Such a reduction transforms an input $(x,y)$ of the function $f$ into a sequence of $2t$ monoid elements $m_1,m_2,...,m_{2t}$ where the odd-indexed $m_i$ are obtained as a function of $x$ only and the even-indexed $m_i$ are a function of $y$.
\end{definition}

We write $f \leq_r^t (M,I)$ to indicate that $f$ has a rectangular reduction of length $t$ to $(M,I)$. When $t = O(n)$ we omit the superscript $t$. It should be clear that if $f \leq_r^t (M,I)$ and $f$ has communication complexity $\Omega(g(n))$, then $(M,I)$ has communication complexity $\Omega(g(t^{-1}(n)))$.

Most of the reductions we use here are special kinds of rectangular reductions. We call these reductions \emph{local rectangular reductions}. In a local rectangular reduction, Alice converts each bit $x_i$ to a sequence of $s$ monoid elements $m_{i,1}^a , m_{i,2}^a , ... , m_{i,s}^a$ by applying a fixed function $a: \{0,1\} \to M^s$. Similarly Bob converts each bit $y_i$ to a sequence of $s$ monoid elements $m_{i,1}^b , m_{i,2}^b , ... , m_{i,s}^b$ by applying a fixed function $b: \{0,1\} \to M^s$. $f(x,y) = 1$ if and only if
\[ m_{1,1}^a m_{1,1}^b ... m_{1,s}^a m_{1,s}^b ...... m_{n,1}^a m_{n,1}^b ... m_{n,s}^a m_{n,s}^b \in I \]
We often view the above product as a word over $M$. The reduction transforms an input $(x,y)$ into a sequence of $2sn$ monoid elements. Let $a(z)_k$ denote the $k^{\textrm{th}}$ coordinate of the tuple $a(z)$. We specify this kind of local transformation with a $2 \times 2s$ matrix:
\[
\begin{array}{|c|c|c|c|c|c|}
\hline
a(0)_1 & b(0)_1 & ... & ... & a(0)_s & b(0)_s \\
\hline
a(1)_1 & b(1)_1 & ... & ... & a(1)_s & b(1)_s \\
\hline
\end{array}.
\]
It is convenient to see what happens for all possible values of $x_i$ and $y_i$ and the following table shows the word that corresponds to these possibilities. For simplicity let us assume $s$ is even.
\[
\begin{array}{|c c|c|}
\hline
x_i & y_i & \textrm{corresponding word} \\
\hline
0 & 0 & a(0)_1 b(0)_1 ... a(0)_s  b(0)_s \\
\hline
0 & 1 & a(0)_1 b(1)_1 a(0)_2 b(1)_2 ... a(0)_s b(1)_s \\
\hline
1 & 0 & a(1)_1 b(0)_1 a(1)_2 b(0)_2 ... a(1)_s b(1)_s \\
\hline
1 & 1 & a(1)_1 b(1)_1 ... a(1)_s b(1)_s \\
\hline
\end{array}
\]

Now we are ready to present the upper and lower bound results.

\begin{lemma}\label{commutativeupperbound}
If $M$ is commutative then $N^1(M) = O(1)$.
\end{lemma}
\begin{proof}
Let $I$ be an order ideal in $M$. Suppose Alice is given $m_1^a,...,m_n^a$ and Bob is given $m_1^b,...,m_n^b$. They want to decide if $m_1^a m_1^b ... m_n^a m_n^b \in I$. Since $M$ is commutative, this is equivalent to determining if $m_1^a m_2^a ... m_n^a m_1^b m_2^b ... m_n^b \in I$. So Alice can privately compute the product $m_1^a ... m_n^a$ and send the result $m$ to Bob. Observe that this requires a constant number of bits to be communicated since the size of $M$ is a constant. Bob can check if $m m_1^b ... m_n^b \in I$ and send the outcome to Alice.
\end{proof}

\begin{lemma}
If $M$ is not commutative then for any order on $M$ we have $N^1(M) = \Omega(\log n)$.
\end{lemma}
\begin{proof}
Since $M$ is not commutative, there must be $a,b \in M$ such that $ab \neq ba$. Therefore either $ab \nleq_M ba$ or $ba \nleq_M ab$. Without loss of generality assume $ba \nleq_M ab$. Let $I = \langle ab \rangle$. We show that $LT \leq_r^{2^n} (M, I)$. Alice gets $x$ and constructs a sequence of $2^n$ monoid elements in which $a$ is in position $x$ and $1_M$ is in everywhere else. Bob gets $y$ and constructs a sequence of $2^n$ monoid elements in which $b$ is in position $y$ and $1_M$ is in everywhere else. If $x \leq y$ then the product of the monoid elements will be $ab$ which is in $I$. If $x > y$ then the product will be $ba$ which is not in $I$.
\end{proof} 

Denote by $\mathcal{C}\it{om}$ the positive language variety corresponding to the variety of commutative monoids $\mathbf{Com}$. The above two results show that regular languages that have constant non-deterministic communication complexity are exactly those languages in $\mathcal{C}\it{om}$.

The next step is to determine if there are regular languages outside of $\mathcal{C}\it{om}$ that have $O(\log n)$ non-deterministic complexity. For this, we first need the definition of a polynomial closure.

The \index{polynomial closure} \emph{polynomial closure} of a set of languages $\mathcal{L}$ in $\Sigma^*$ is a family of languages such that each of these languages are finite unions of languages of the form
\[ L_0 a_1 L_1 ... a_k L_k \]
where $k \geq 0$, $a_i \in \Sigma$ and $L_i \in \mathcal{L}$. If $\mathcal{V}$ is a variety of languages, then we denote by $Pol(\mathcal{V})$ the class of languages such that for every alphabet $\Sigma$, $Pol(\mathcal{V})(\Sigma)$ is the polynomial closure of $\mathcal{V}(\Sigma)$. $Pol(\mathcal{V})$ is a positive variety of languages (\cite{PW95}).

\begin{lemma}
If $L$ is a language of $Pol(\mathcal{C}\it{om})$ then $N^1(L) = O(\log n)$.
\end{lemma}
\begin{proof}
Suppose $L$ is a union of $t$ languages of the form $L_0 a_1 L_1 ... a_k L_k$. Alice and Bob know beforehand the value of $t$ and the structure of each of these $t$ languages. So a protocol for $L$ is as follows.

Assume Alice is given $x_1^a ,..., x_n^a$ and Bob is given $x_1^b , ... , x_n^b$. God communicates to Alice and Bob which of the $t$ languages the word $x_1^a x_1^b ... x_n^a x_n^b$ resides in. This requires a constant number of bits to be communicated since $t$ is a constant. Now that Alice and Bob know the $L_0 a_1 L_1 ... a_k L_k$ the word is in, God communicates the positions of each $a_i$. This requires $k \log n$ bits of communication where $k$ is a constant. The validity of the information communicated by God can be immediately checked by Alice and Bob. All they have to do is check if the words in between the $a_i$'s belong to the right languages. Since these languages are in $\mathcal{C}\it{om}$, this can be done in constant communication as proved in Lemma \ref{commutativeupperbound}. Therefore in total we require only $O(\log n)$ communication. 
\end{proof}

From the above proof, we see that we can actually afford to communicate $O(\log n)$ bits to check that the words between the $a_i$'s belong to the corresponding language. In other words, we could have $L_i \in Pol(\mathcal{C}\it{om})$. Note that this does not matter since $Pol(Pol(\mathcal{C}\it{om})) = Pol(\mathcal{C}\it{om})$.


Denote by ${x \choose L_0 a_1 L_1 ... a_k L_k}$ the number of factorizations of the word $x$ as $x = w_0 a_1 w_1 ... a_k w_k$ with $w_i \in L_i$. When the $a_i$ and the $L_i$ are such that for any $x$ we have ${x \choose L_0 a_1 L_1 ... a_k L_k} \in \{0,1\}$, then we say that the concatenation $L_0 a_1 L_1 ... a_k L_k$ is \emph{unambiguous}. We denote by $UPol(\mathcal{V})$ the variety of languages that is \emph{disjoint} unions of the unambiguous concatenations $L_0 a_1 L_1 ... a_k L_k$ with $L_i \in \mathcal{V}$ (in some sense, there is only one witness for $x$ in $UPol(\mathcal{V})$). Similarly we denote by $M_pPol(\mathcal{V})$ the language variety generated by the languages
\[
\{ x | {x \choose L_0 a_1 L_1 ... a_k L_k} = j \mod p \}
\]
for some $0 \leq j \leq p-1$ and $L_i \in \mathcal{V}$. Observe that for $Pol(\mathcal{C}\it{om})$ we have ${x \choose L_0 a_1 L_1 ... a_k L_k}$ unrestricted.

Denote by $UP$ the subclass of $NP$ in which the number of accepting paths (or number of witnesses) is exactly one. We know that $UP^{cc} = P^{cc}$ (\cite{Yan91}). From $\cite{TT03}$ we know that regular languages having $O(\log n)$ deterministic communication complexity are exactly those languages in $UPol(\mathcal{C}\it{om})$ and regular languages having $O(\log n)$ Mod$_p$ counting communication complexity are exactly those languages in $M_pPol(\mathcal{C}\it{om})$. Furthermore, it was shown that any regular language outside of $UPol(\mathcal{C}\it{om})$ has linear deterministic complexity and any regular language outside of $M_pPol(\mathcal{C}\it{om})$ has linear Mod$_p$ counting complexity. So with respect to regular languages, $UP^{cc} = P^{cc} = UPol(\mathcal{C}\it{om})$ and $Mod_pP^{cc} = M_pPol(\mathcal{C}\it{om})$.  Similarly we conjecture that with respect to regular languages $NP^{cc} = Pol(\mathcal{C}\it{om})$. 

\begin{conjecture}\label{conj}
If $L \subseteq \Sigma^*$ is a regular language that is not in $Pol(\mathcal{C}\it{om})$, then $N^1(L) = \Omega(n)$. Thus we have
\[
N^1(L) = \left\{
          \begin{array}{ll}
          O(1) & \quad \textrm{if and only if $L \in \mathcal{C}\it{om}$;}  \\
          \Theta(\log n) & \quad \textrm{if and only if $L \in Pol(\mathcal{C}\it{om})$ but not in $\mathcal{C}\it{om}$;} \\
	  \Theta(n) & \quad \textrm{otherwise.}
          \end{array}
          \right.
\]
\end{conjecture}

As mentioned in Chapter 2, the gap between deterministic and non-deterministic communication complexity of a function can be exponentially large. However, it has been shown that the deterministic communication complexity of a function $f$ is bounded above by the product $c N^0(f) N^1(f)$ for a constant $c$ (Theorem \ref{detvsnondet}), and that this bound is optimal. The above conjecture, together with the result of \cite{TT03} implies the following much tighter relation for regular languages.

\begin{corollary}[to Conjecture \ref{conj}]
If $L$ is a regular language then $D(L) = \max \{N^1(L), N^0(L)\}$.
\end{corollary}

For any variety $\mathcal{V}$, we have that $Pol(\mathcal{V}) \cap co-Pol(\mathcal{V}) = UPol(\mathcal{V})$ (\cite{Pin97}). This implies that $N^1(L) = O(\log n)$ and $N^0(L) = O(\log n)$ iff $D(L) = O(\log n)$, proving a special case of the above corollary.

An important question that arises in this context is the following. What does it mean to be outside of $Pol(\mathcal{C}\it{om})$? In order to prove a linear lower bound for the regular languages outside of $Pol(\mathcal{C}\it{om})$, we need a convenient algebraic description for the syntactic monoids of these languages since (ignoring the exceptions) lower bound arguments rely on these algebraic properties. One such description exists based on a result of \cite{PW95} that describes the ordered monoid variety corresponding to $Pol(\mathcal{C}\it{om})$. Before stating this description, we fix some notation.

If $M$ is a monoid, we write $M = \langle G,R\rangle$ to indicate that $M$ has the presentation $\langle G,R\rangle$ where $G$ is the generating set and $R$ is the set of relations. For instance, a cyclic group of order $n$ has the presentation $\langle \{x\}, x^n = 1 \rangle$ and the dihedral group of order $2n$ has the presentation $\langle \{x,y\}, x^n = 1, y^2 = 1, xyx = y \rangle$. For any $w \in G^*$, we denote by $eval(w)$ the element of $M$ that $w$ corresponds to. Observe that the transformation monoid corresponding to an automaton has a presentation in which the generating set consists of the letters of the alphabet. The relations depend on the particular automaton and can be determined by analyzing the state transition function each word induces. 


\begin{lemma}\label{description}
Suppose $L$ is not in $Pol(\mathcal{C}\it{om})$ and $M = \langle G,R \rangle$ is the syntactic ordered monoid of $L$ with exponent $\omega$. Then there exists $u,v \in G^*$ such that
\begin{enumerate}
\item[\emph{(i)}] for any monoid $M' \in \mathbf{Com}$ and any morphism $\phi: M \to M'$, we have $\phi(eval(u)) = \phi(eval(v))$ and $\phi(eval(u)) = \phi(eval(u^2))$,
\item[\emph{(ii)}] $eval(u^\omega v u^\omega) \nleq eval(u^\omega)$.
\end{enumerate}
\end{lemma}

Although we cannot yet prove the conjecture, we can still show linear lower bounds for certain classes of regular languages outside of $Pol(\mathcal{C}\it{om})$. Our first lower bound captures regular languages that come very close to the description given in the previous lemma.

A word $w$ is a \index{shuffle} \emph{shuffle} of $n$ words $w_1,...,w_n$ if 
\[ w = w_{1,1}w_{2,1}...w_{n,1}w_{2,1}w_{2,2}...w_{n,2}w_{1,k}w_{2,k}...w_{n,k} \]
with $k \geq 0$ and $w_{i,1}w_{i,2}...w_{i,k} = w_i$ is a partition of $w_i$ into subwords for $1 \leq i \leq n$.

\begin{lemma}\label{shuffle}
If $M = \langle G,R \rangle$ and $u,v \in G^*$ is such that 
\begin{enumerate}
\item[\emph{(i)}] $u = w_1w_2$ for $w_1,w_2 \in G^*$,
\item[\emph{(ii)}] $v$ is a shuffle of $w_1$ and $w_2$,
\item[\emph{(iii)}] $eval(u)$ is an idempotent,
\item[\emph{(iv)}] $eval(u v u) \nleq eval(u)$,
\end{enumerate}
then $N^1(M) = \Omega(n)$.
\end{lemma}

Observe that the conditions of this lemma imply the conditions of Lemma \ref{description}: since $eval(u)$ is idempotent, for any monoid $M' \in \mathbf{Com}$ and any morphism $\phi:M \to M'$, we have $\phi(eval(u)) = \phi(eval(u^2))$ and since $v$ is a shuffle of $w_1$ and $w_2$ we have $\phi(eval(u)) = \phi(eval(v))$. Also, since $eval(u)$ is idempotent, $eval(u^\omega) = eval(u)$, and in this case $eval(u v u) \nleq eval(u)$ is equivalent to $eval(u^\omega v u^\omega) \nleq eval(u^\omega)$.

\begin{proof}[Proof of Lemma \ref{shuffle}]
We show that $PDISJ \leq_r (M,I)$ where $I = \langle eval(u) \rangle$. Since $v$ is a shuffle of $w_1$ and $w_2$, there exists $k \geq 0$ such that 
\[v = w_{1,1}w_{2,1}w_{1,2}w_{2,2}...w_{1,k}w_{2,k}. \] 
The reduction is essentially linear and is given by the following matrix when $k=3$. The transformation easily generalizes to any $k$.
\[
\begin{array}{|c|c|c|c|c|c|c|c|c|c|}
\hline
w_1 & \epsilon & \epsilon & \epsilon & \epsilon & w_{2,1} & \epsilon & w_{2,2} & \epsilon & w_{2,3} \\
\hline
w_{1,1} & w_{2,1} & w_{1,2} & w_{2,2} & w_{1,3} & w_{2,3} & \epsilon & \epsilon & \epsilon & \epsilon \\
\hline
\end{array}.
\]
\[
\begin{array}{|c c|c|}
\hline
x_i & y_i & \textrm{corresponding word} \\
\hline
0 & 0 & w_1 w_{2,1} w_{2,2} w_{2,3} = u \\
\hline
0 & 1 & w_1 w_{2,1} w_{2,2} w_{2,3} = u \\
\hline
1 & 0 & w_{1,1} w_{1,2} w_{1,3} w_{2,1} w_{2,2} w_{2,3} = u\\
\hline
1 & 1 & w_{1,1} w_{2,1} w_{1,2} w_{2,2} w_{1,3} w{2,3} = v \\
\hline
\end{array}
\]
After $x$ and $y$ have been transformed into words, Alice prepends her word with $u$ and appends it with $|u|$ many $\epsilon$'s, where $|u|$ denotes the length of the word $u$. Bob prepends his word with $|u|$ many $\epsilon$'s and appends it with $u$. Let $a(x)$ be the word Alice has and let $b(y)$ be the word Bob has after these transformations. Observe that if $PDISJ(x,y) = 0$, then there exists $i$ such that $x_i = y_i = 1$. By the transformation, this means that $a(x)_1b(x)_1a(x)_2b(x)_2 ... a(x)_sb(x)_s$ is of the form $u...uvu...u$ and since $eval(u)$ is idempotent, $eval(a(x)_1b(x)_1a(x)_2b(x)_2 ... a(x)_sb(x)_s) = eval(u v u) \nleq eval(u)$. On the other hand if $PDISJ(x,y) = 1$, then by the transformation, $a(x)_1b(x)_1a(x)_2b(x)_2 ... a(x)_sb(x)_s$ is of the form $u...u$ and so 
\[ eval(a(x)_1b(x)_1a(x)_2b(x)_2 ... a(x)_sb(x)_s) = eval(u). \]
\end{proof}

The above result gives us a corollary about the monoid $BA_2^+$ which is defined to be the syntactic ordered monoid of the regular language recognized by the automaton in Figure \ref{bicycle}. The unordered syntactic monoid of the same language is denoted by $BA_2$ and is known as the Brandt monoid (see \cite{Pin97}).

\begin{figure}
\begin{center}
\includegraphics[scale=0.8]{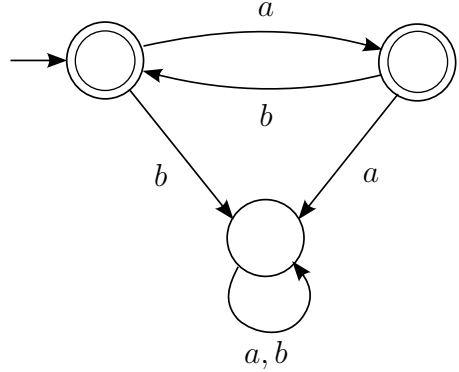}
\rput{0}(-4.1,4.8){$a$}
\rput{0}(-4.1,3.4){$b$}
\rput{0}(-5.5,2.6){$b$}
\rput{0}(-2.7,2.6){$a$}
\rput{0}(-4.1,0.2){$a,b$}
\caption{The minimal automaton recognizing the language whose syntactic ordered monoid is $BA_2^+$.}
\label{bicycle}
\end{center}
\end{figure}

\begin{corollary}\label{corbi}
$N^1(BA_2^+) = \Omega(n)$.
\end{corollary}
\begin{proof}
It is easy to verify by looking at the transformation monoid of the automaton that $BA_2^+ = \langle \{a,b\}, aa = bb, aab = aa, baa = aa, aaa = a, aba=a, bab=b \rangle$. The only thing we need to know about the order relation is that $eval(aa)$ is greater than any other element. This can be derived from the definition of the syntactic ordered monoid (Subsection 3.2.2) since for any $w_1$ and $w_2$, $w_1aaw_2$ is not in $L$. So $w_1aaw_2 \in L \implies w_1xw_2 \in L$ trivially holds for any word $x$. Let $u = ab$ and $v = ba$. These $u$ and $v$ satisfy the four conditions of the previous lemma. The last condition is satisfied because $eval(uvu) = eval(abbaab) = eval(aa)$ and $eval(ab) \neq eval(aa)$. Therefore $N^1(BA_2^+) = \Omega(n)$.
\end{proof}

Denote by $U^-$ the syntactic ordered monoid of the regular language $(a \cup b)^*aa(a \cup b)^*$, and denote by $U$ the unordered syntactic monoid. Also let $U^+$ be the syntactic ordered monoid of the complement of $(a \cup b)^*aa(a \cup b)^*$. Observe that $N^1(U^-) = O(\log n)$ since all we need to do is check if there are two consecutive $a$'s. By an argument similar to the one for Corollary \ref{corbi}, one can show that $N^1(U^+) = \Omega(n)$.

Our next linear lower bound result is for non-commutative groups.

\begin{lemma}
If $M$ is a non-commutative group then $N^1(M) = \Omega(n)$.
\end{lemma}
\begin{proof}
Since $M$ is non-commutative, there exists $a,b \in M$ such that the commutator $[a,b] = a^{-1}b^{-1}ab \neq 1$. This means that $[a,b]$ has order $q > 1$. Let $m \in M$ be such that there is no $m' \in M$ with $m' \neq m$ and $m' \leq m$. Denote by $I$ the order ideal that just contains $m$. There is a reduction from $IP_q$ to $(M,I)$. The reduction is essentially local. Alice and Bob will apply the transformation given by the following matrix.
\[
\begin{array}{|c|c|c|c|}
\hline
1 & 1 & 1 & 1 \\
\hline
a^{-1} & b^{-1} & a & b \\
\hline
\end{array}.
\]
\[
\begin{array}{|c c|c|}
\hline
x_i & y_i & \textrm{corresponding word} \\
\hline
0 & 0 & 1 \\
\hline
0 & 1 & b^{-1}b = 1 \\
\hline
1 & 0 & a^{-1}a = 1 \\
\hline
1 & 1 & a^{-1} b^{-1} a b \\
\hline
\end{array}
\]
After, Alice will append $m$ to her transformed input and Bob will append $1$ to his. Observe that the product of the monoid elements evaluates to $m$ if and only if $\sum_{1 \leq i \leq n} x_iy_i \equiv 0 \mod q$ i.e. the product is $\leq m$ if and only if $\sum_{1 \leq i \leq n} x_iy_i \equiv 0 \mod q$.
\end{proof}

To obtain our last linear lower bound result, we need the following fact.

\begin{proposition}
Any stable order defined on a group $G$ must be the trivial order (equality).
\end{proposition}
\begin{proof}
Suppose the claim is false. So there exists $a,b \in G$ such that $a \neq b$ and $a \leq b$. This implies $1 \leq a^{-1}b =: g$. If $1 \leq g$ then $g \leq g^2$, $g^2 \leq g^3$ and so on. Therefore we have $1 \leq g \leq g^2 \leq ... \leq g^k = 1$. This can only be true if $1 = g$, i.e. $a = b$.
\end{proof}

We say that $M$ is a $T_q$ monoid if there exists idempotents $e,f \in M$ such that $(ef)^qe = e$ but $(ef)^re \neq e$ when $q$ does not divide $r$.

\begin{lemma}
If $M$ is a $T_q$ monoid for $q > 1$ then $N^1(M) = \Omega(n)$.
\end{lemma}
\begin{proof}
Observe that $\{e, efe, (ef)^2e, ... , (ef)^{q-1}e\}$ forms a subgroup with identity $e$ because since $e$ is idempotent, we have $(ef)^ie \cdot (ef)^je = (ef)^{i+j}e$. Therefore any order on $M$ must induce an equality order on this set. Let $I = \langle e \rangle$. We show $IP_q \leq_r (M,I)$ via the following local reduction.
\[
\begin{array}{|c|c|}
\hline
e(ef)^q & (ef)^q e \\
\hline
e & fe \\
\hline
\end{array}.
\]
\[
\begin{array}{|c c|c|}
\hline
x_i & y_i & \textrm{corresponding word} \\
\hline
0 & 0 & e(ef)^q(ef)^qe = e \\
\hline
0 & 1 & e(ef)^qfe = e \\
\hline
1 & 0 & e(ef)^qe = e \\
\hline
1 & 1 & efe \\
\hline
\end{array}
\]
Observe that the product of the monoid elements evaluates to 
\[(ef)^{\sum_{1\leq i \leq n} x_iy_i \equiv 0 \mod q} e, \] 
which is equal to $e$ if and only if $IP_q(x,y) = 1$.
\end{proof}

Combining our linear lower bound results together with Lemma \ref{division}, we can conclude the following.

\begin{theorem}\label{}
If $M$ is a $T_q$ monoid for $q > 1$ or is divided by one of $BA_2^+$, $U^+$ or a non-commutative group, then $N^1(M) = \Omega(n)$.
\end{theorem}

We underline the relevance of the above result by stating a theorem which we borrow from \cite{TT05}.

\begin{theorem}\label{deterministic}
If $M$ is such that $D(M) \neq O(\log n)$ then $M$ is either a $T_q$ monoid for some $q > 1$ or is divided by one of $BA_2$, $U$ or a non-commutative group.
\end{theorem}

The three linear lower bound results imply the following result, which gives us three sufficient conditions for not being in $Pol(\mathcal{C}\it{om})$.

\begin{theorem}
Let $L$ be a regular language with syntactic ordered monoid $M = \langle G, R \rangle$. If one of the following holds, then $L$ is not in $Pol(\mathcal{C}\it{om})$.
\begin{enumerate}
\item There exists $u,v \in G^*$ such that $u = w_1w_2$, $v$ is a shuffle of $w_1$ and $w_2$, $eval(u)$ is an idempotent and $eval(uvu) \nleq eval(u)$.
\item $M$ is divided by a non-commutative group.
\item $M$ is a $T_q$ monoid for $q > 1$.
\end{enumerate}
\end{theorem}

In particular, if $M(L)$ is a $T_q$ monoid or is divided by one of $BA_2^+$, $U^+$ or a non-commutative group, then $L$ is not in $Pol(\mathcal{C}\it{om})$.

%% file: chap5.tex
\chapter{Conclusion}\label{chapter5}

The focus of this thesis has been the non-deterministic communication complexity of regular languages. Regular languages are, in some sense, the simplest languages with respect to the usual time/space complexity framework, but in the communication complexity model, they require a non-trivial study as there are complete regular languages for every level of the communication complexity polynomial hierarchy. This fact can be derived from the results in \cite{Bar86} and \cite{BT87}. In \cite{TT03}, a complete characterization of the communication complexity of regular languages was established in the deterministic, simultaneous, probabilistic, simultaneous probabilistic and Mod$_p$-counting models. In order to get a similar algebraic characterization for the non-deterministic model, one needs the notion of ordered monoids, a more general theory than the one used in \cite{TT03}, to be able to deal with classes of languages that are not closed under complementation. This thesis presents the fundamentals of communication complexity, monoid theory as well as ordered monoid theory and obtains bounds on the non-deterministic communication complexity of regular languages.

Our results constitute the first steps towards a complete classification for the non-deterministic communication complexity of regular languages. We know exactly which regular languages have constant non-deterministic communication complexity. We know that there is a considerable complexity gap between those languages having constant non-deterministic complexity and the rest of the regular languages since if a regular language does not have constant complexity than it has $\Omega(\log n)$ complexity. We also obtain three linear lower bound results and the importance of these results are highlighted by Theorem \ref{deterministic}. These results also provide us with several sufficient conditions for not being in the variety $Pol(\mathcal{C}\it{om})$, which is a result very interesting from an algebraic automata theory point of view. This result also exemplifies how computational complexity can be used to make progress in semigroup theory.

Our ultimate objective is to get a complete characterization of the non-deterministic communication complexity of regular languages. We conjecture that regular languages in $Pol(\mathcal{C}\it{om})$ are the only languages that have $O(\log n)$ complexity and any other regular language must have $\Omega(n)$ complexity. The linear lower bound argument presents a real challenge. A natural next step to take is to explicitly find a regular language that is not in $Pol(\mathcal{C}\it{om})$ for which our current linear lower bound arguments do not apply and try to either prove a linear lower bound for this specific language or show that it requires $O(n^\epsilon)$ complexity for a constant $\epsilon < 1$ (which would disprove our conjecture). A linear lower bound argument for this language is likely to apply to some other languages outside of $Pol(\mathcal{C}\it{om})$, if not all. The regular language recognized by the automaton in Figure \ref{automaton} is an example of a regular language that is outside of $Pol(\mathcal{C}\it{om})$ and for which we cannot get a linear lower bound nor a sublinear upper bound. We call this language $L_5$. In the Appendix, we prove that $L_5$ is not in $Pol(\mathcal{C}\it{om})$ and various other facts about $L_5$.

\begin{figure}
\begin{center}
\includegraphics[scale=0.8]{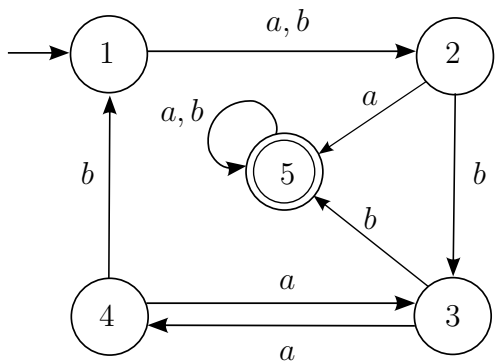}
\rput{0}(-4,4.9){$a,b$}
\rput{0}(-4,1.45){$a$}
\rput{0}(-4,0.5){$a$}
\rput{0}(-1.45,2.9){$b$}
\rput{0}(-6.66,2.9){$b$}
\rput{0}(-5.4,3.7){$a,b$}
\rput{0}(-2.9,3.9){$a$}
\rput{0}(-2.9,2.3){$b$}
\rput{0}(-6.4,4.5){1}
\rput{0}(-1.8,4.5){2}
\rput{0}(-1.8,1){3}
\rput{0}(-6.4,1){4}
\rput{0}(-4,2.9){5}
\end{center}
\caption{An automaton recognizing a language $L_5$ outside of $Pol(\mathcal{C}\it{om})$.}
\label{automaton}
\end{figure}

An interesting property of $L_5$ is that it belongs to $Pol(\mathcal{N}\it{il}_2)$ where $\mathcal{N}\it{il}_2$ denotes the variety of languages that correspond to the variety of nilpotent groups of class 2. Nilpotent groups of class 2 are usually considered as ``almost" commutative groups. In some sense, this says that even though $L_5$ is not in $Pol(\mathcal{C}\it{om})$, it is very ``close" to it.

We propose several intuitive reasons of why proving a linear lower bound for this regular language can be challenging (assuming that the linear lower bound is indeed true). We also suggest possible approaches to overcome the difficulties. All of these tie with the importance of the problem we are studying from a communication complexity point of view as well as from a semigroup theory point of view.

First of all, from Chapter 2 we know that the best lower bound technique we have for non-determinism is the rectangle size method. Inherent in this method is the requirement to find the best possible distribution. Needless to say, this can be quite hard. And even if the best distribution was known, bounding the size of any 1-monochromatic rectangle can be a non-trivial task. Putting these two things together, the rectangle size method does not seem to considerably simplify our task of bounding the size of the optimum covering of the 1-inputs.

Consider the set $S$ of all functions having $\Omega(n)$ non-deterministic communication complexity. Define an equivalence relation on these functions: $f \equiv g$ if there is a rectangular reduction of length $O(n)$ from $f$ to $g$ and from $g$ to $f$. We can turn $S/\equiv$ into a partially ordered set (poset) by defining the order $[f] < [g]$ if there is a rectangular reduction of length $O(n)$ from $f$ to $g$. It certainly would not be surprising if there were regular languages appearing in the lower levels of a chain in this poset and this would suggest that obtaining a lower bound for these languages can be difficult. 

If the above is indeed true, then what can be done about this? A natural step would be to find functions that are at the same level or below the regular language at hand, and try to get a reduction that would prove the language has linear non-deterministic complexity. This raises our interest in promise functions.

Let $f$ be a boolean function with the domain $\{0,1\}^n \times \{0,1\}^n$. A promise function $Pf$ is a function that has a domain $D$ that is a strict subset of $f$'s domain and is such that for any $(x,y) \in D$, $Pf(x,y) = f(x,y)$. An example of a promise function is the PROMISE-DISJOINTNESS function, $PDISJ$. Promise functions are interesting because through a promise, we can define functions that reside in the lower levels of a chain. This in return can make a reduction possible from the promise function to the regular language of interest. For instance, $PDISJ$ is a promise function which lies below $DISJ$ and $IP_q$ (Example \ref{reduction}). Of course an important point when defining a promise function is that we need the promise function to have $\Omega(n)$ complexity. In some sense, through the promise, we would like to eliminate the easy instances and keep the instances that make the function hard. At first, there might be no reason to believe that obtaining a linear lower bound for the promise function is any easier than obtaining a lower bound for the regular language. Nevertheless, the purpose of this line of attack is the following. By putting a promise on a well-known, well-studied function (that makes a reduction possible), we may be able to utilize (or improve) the various techniques and ideas developed for the analysis of the original function to prove a lower bound on the promise function.

\index{PROMISE-INNER-PRODUCT}
Now we define a promise function, PROMISE-INNER-PRODUCT ($PIP_2$), such that there is a reduction from this function to $L_5$ (see Appendix). $PIP_2$ is the same function as $IP_2$ but has a restriction on the $(x,y)$ for which $IP_2(x,y) = 0$. We only allow the 0-inputs which satisfy the following two conditions:
\[ \forall i, \; x_i = 0 \textrm{ and } y_i = 1 \implies IP_2(x_1...x_{i-1}, y_1...y_{i-1}) = 0 \]
\[ \forall i, \; x_i = 1 \textrm{ and } y_i = 0 \implies IP_2(x_1...x_{i-1}, y_1...y_{i-1}) = 1 \]
It remains an open problem to prove a linear lower bound, or a sublinear upper bound on $PIP_2$.

The fact is that little is known about promise functions. One promise function we know of is $PDISJ$. As a consequence of the celebrated work of Razborov (\cite{Raz92}), which shows that the distributional communication complexity of the DISJOINTNESS function is $\Omega(n)$, we know that $N^1(PDISJ) = \Omega(n)$ as well. Given the description of regular languages outside of $Pol(\mathcal{C}\it{om})$ (Lemma \ref{description}), $PDISJ$ is one of the first functions one tries to get a reduction from, where the reduction is as in the proof of Lemma \ref{shuffle}. This hope is hurt by the fact that such a reduction does not exist from $PDISJ$ to $L_5$ (see Appendix).

We believe that more attention should be given to promise functions since the study of these functions is likely to force us to develop new techniques in communication complexity and give us more insight in this area. Furthermore, given the connection of communication complexity with many other areas in computer science, promise functions are bound to have useful applications. For instance, in a very recent work of G\'{a}l and Gopalan (\cite{GG07}), communication complexity bounds for a promise function is used to prove bounds on streaming algorithms.

We have looked at our question from a communication complexity perspective. Now we look at it from a semigroup theory perspective. The key to making progress on our question can be finding a more convenient description of what it means to be outside of $Pol(\mathcal{C}\it{om})$. The description that we have (Lemma \ref{description}) actually applies to $Pol(\mathcal{V})$ for any variety $\mathcal{V}$, and it is based on a complicated result of \cite{PW95} that makes use of a deep combinatorial result of semigroup theory (\cite{Sim89},\cite{Sim90},\cite{Sim92}). Since we are only interested in $Pol(\mathcal{C}\it{om})$ where $\mathcal{C}\it{om}$ is a relatively simple variety, it may be possible to obtain a more useful description that allows us to show communication complexity bounds.

We conclude that, in any case, the resolution of our question will probably lead to advances in either communication complexity or semigroup theory, if not both.

%% file: app.tex
\chapter{Facts About $L_5$}

\begin{center}
\includegraphics[scale=0.8]{automaton.eps}
\rput{0}(-4,4.9){$a,b$}
\rput{0}(-4,1.45){$a$}
\rput{0}(-4,0.5){$a$}
\rput{0}(-1.45,2.9){$b$}
\rput{0}(-6.66,2.9){$b$}
\rput{0}(-5.4,3.7){$a,b$}
\rput{0}(-2.9,3.9){$a$}
\rput{0}(-2.9,2.3){$b$}
\rput{0}(-6.4,4.5){1}
\rput{0}(-1.8,4.5){2}
\rput{0}(-1.8,1){3}
\rput{0}(-6.4,1){4}
\rput{0}(-4,2.9){5}
\end{center}

In this appendix, we prove some of the facts about the regular language $L_5$ that we mentioned in Chapter 5. We start with the fact that $L_5$ is not in $Pol(\mathcal{C}\it{om})$. For this, we need a result that describes the ordered monoid variety corresponding to $Pol(\mathcal{C}\it{om})$. This description involves the Mal'cev product and topological issues which we choose to avoid for simplicity. The interested reader can find the necessary information about these in \cite{Pin97}. Here we will state a restricted version of this result which suffices for our needs.

\begin{lemma}
Let $L$ be a language in $Pol(\mathcal{C}\it{om})$ and let $M = \langle G, R \rangle$ be the syntactic ordered monoid of $L$ with exponent $\omega$. Then for any $u, v \in G^*$ with the property that any monoid $M' \in \mathbf{Com}$ and any morphism $\phi: M \to M'$ satisfies both $\phi(eval(u)) = \phi(eval(v))$ and $\phi(eval(u)) = \phi(eval(u^2))$, we must have $eval(u^\omega vu^\omega) \leq eval(u^\omega)$.
\end{lemma}

\begin{proposition}
$L_5$ is not in $Pol(\mathcal{C}\it{om})$.
\end{proposition}
\begin{proof}
Consider the transformation monoid of $L_5$, which is the syntactic monoid. Let $u = abab$ and $v = bbaa$. Observe that $eval(u)$ is an idempotent and this $u$ and $v$ satisfy the condition in the lemma. We show $eval(uvu) \nleq eval(u)$. Observe that $eval(uvu) = eval(v)$ so we want to show $eval(v) \nleq eval(u)$. If the opposite was true, then by the definition of the syntactic ordered monoid (Subsection 3.2.2), we must have for any $w_1$ and $w_2$, $w_1uw_2 \in L \implies w_1vw_2 \in L$. In particular, for $w_1 = \epsilon$ and $w_2 = aa$, we would have $uaa \in L \implies vaa \in L$. It is true that $uaa \in L$ but $vaa \notin L$.
\end{proof}

Now we show that the PROMISE-INNER-PRODUCT function that we defined in Chapter 5 reduces to $L_5$.

\begin{proposition}
$PIP_2 \leq L_5$
\end{proposition}
\begin{proof}
The reduction is linear and is given by the following matrix.
\[
\begin{array}{|c|c|c|c|c|c|c|c|}
\hline
a & \epsilon & \epsilon & b & a & b & \epsilon & \epsilon \\
\hline
a & b & \epsilon & a & b & a & \epsilon & b \\
\hline
\end{array}.
\]
\[
\begin{array}{|c c|c|c|}
\hline
x_i & y_i & \textrm{corresponding word} & \textrm{state transition} \\
\hline
0 & 0 & abab & 1 \to 1, \; 2 \to 5, \; 3 \to 3, \; 4 \to 5, \; 5 \to 5 \\
0 & 1 & abaaab & 1 \to 1, \; 2 \to 5, \; 3 \to 5, \; 4 \to 5, \; 5 \to 5\\
1 & 0 & abbb &  1 \to 5, \; 2 \to 5, \; 3 \to 3, \; 4 \to 5, \; 5 \to 5 \\
1 & 1 & ababab & 1 \to 3, \; 2 \to 5, \; 3 \to 1, \; 4 \to 5, \; 5 \to 5 \\
\hline
\end{array}.
\]
After this transformation is applied, Alice appends her word with $b$ and Bob appends his word with $\epsilon$.  Now if $PIP_2(x,y) = 1$ then the transformed word must end up at state 5. This is because, from state 1, we either enter state 5 and stay there forever, or we go to state 3 when $x_i = y_i = 1$. If we were in state 3 already, $x_i = y_i = 1$ takes us back to state 1. $PIP_2(x,y)=1$ implies that there are an odd number of indices for which $x_i = y_i = 1$. So after the linear transformation, assuming we do not end up in state 5, we would end up in state 3. The $b$ appended at the end of the transformed word would ensure that we end up in state 5.  If $PIP_2(x,y) = 0$, then we do not want to end up at state 5. Observe that the promise ensures we never enter state 5 and since there are even number of indices for which $x_i = y_i = 1$, we must end up at state 1. The $b$ appended at the end of the word just takes us from state 1 to 2.
\end{proof}

Observe that we can restrict the 1-inputs of $PIP_2$ the same way we restricted the 0-inputs and the reduction would trivially work for this case as well. Putting a promise on both the 0-inputs and the 1-inputs may help analyzing the complexity of $PIP_2$.

\begin{proposition}\label{noreduction}
There is no local reduction from $PDISJ$ to $L_5$ such that the reduction is of the form
\[
\begin{array}{|c c|c|}
\hline
x_i & y_i & \textrm{corresponding word} \\
\hline
0 & 0 & u^\omega  \\
0 & 1 & u^\omega \\
1 & 0 & u^\omega  \\
1 & 1 & v  \\
\hline
\end{array}.
\]
where $u$ and $v$ satisfy the conditions of Lemma \ref{description}.
\end{proposition}
\begin{proof}{\it(Sketch).}
Since $u^\omega$ is an idempotent, it must induce a state transition function in which either\\
1. $k \to k$, $j \to j$ and other states are sent to 5, or \\
2. $k \to k$ and other states are sent to 5, or\\
3. every state is sent to 5.\\
Observe that it cannot be the case that $u^\omega$ is a partial identity on more than two states. If $u^\omega$ satisfies condition 2, then we cannot have $eval(u^\omega v u^\omega) \nleq eval(u^\omega)$. Suppose this is the case. Then there exists $w_1,w_2$ such that $w_1 u^\omega w_2 \in L$ and $w_1 u^\omega v u^\omega w_2 \notin L$. Since the latter is true, it must be the case that $w_1$ takes state 1 to $k$ and $w_2$ must take $k$ to a state other than 5. These $w_1$ and $w_2$ do not satisfy $w_1 u^\omega w_2 \in L$, so we get a contradiction. This shows we cannot have condition 2. Similarly, one can show that $u^\omega$ cannot satisfy condition 3, which leaves us with condition 1. This means $u^\omega$ is either $(abab)^k$ or $(baba)^k$ for some $k >0$. We assume it is $(abab)^k$. The argument for $(baba)^k$ is very similar.

Given $u^\omega = (abab)^k$, and the fact that we want to satisfy $eval(u^\omega v u^\omega) \nleq eval(u^\omega)$, one can show that the state transition function induced by $v$ must be one of the following.\\
1. $1 \to 3$ and any other state is sent to 5.\\
2. $3 \to 1$ and any other state is sent to 5.\\
3. $1 \to 3$ and $3 \to 1$ and any other state is sent to 5.\\

Suppose $v$ satisfies condition 1.\\
Case 1: $v = (ab)^{2t-1}$ for $t>0$. Consider the matrix representation of the local reduction. In this matrix $A$, we count the parity of the $a$'s in two ways and get a contradiction. First we count it by looking at the rows. The first row must produce the word $u^\omega = (abab)^k$ and the second row must produce the word $v = (ab)^{2t-1}$ so in total we have odd number of $a$'s. Now we count the parity of $a$'s by looking at $A_{1,1} A_{2,2} A_{1,3} A_{2,4} ...$ and $A_{2,1} A_{1,2} A_{2,3} A_{1,4} ...$. Both of these must produce the word $(abab)^k$ so in total we must have an even number of $a$'s.\\
Case 2: $v = (ab)^{2t}bb...$ for $t \geq 0$. Let $c$ be the column where we find the second $b$ in the second row. Give value 1 to entries of $A$ which are $a$ and give value -1 to entries of $b$. Other entries (the $\epsilon$'s) get value 0. In terms of these values we have 
\[ \sum_{i = 1}^c A_{2,i} = -2\] 
and 
\[\sum_{i = 1}^c A_{1,i} \in \{0,1\}.\] 
Adding the two sums, we get a negative value. Now we count the same total in a different order. Assuming $c$ is even we have \[\sum_{i=1}^{c/2} A_{1,2i - 1} + \sum_{i=1}^{c/2} A_{2,2i} \in \{0,1\}\] 
and
\[ \sum_{i=1}^{c/2} A_{1,2i} + \sum_{i=1}^{c/2} A_{2,2i-1} \in \{0,1\}.\]
The total is positive. This is a contradiction.\\
Case 3: $v = (ab)^{2t-1}a(aa)^{2t}b...$ for $t>0$. Similar argument as above.

Same ideas show that $v$ cannot satisfy neither conditions 2 nor 3.
\end{proof}
